%% file: Main.tex
\documentclass[12pt, a4paper]{article}

\usepackage[T1]{fontenc}
\usepackage[utf8]{inputenc}
\usepackage{lmodern}
\usepackage{titlesec}

\usepackage[a4paper, margin=2.5cm]{geometry}
\usepackage{setspace}
\usepackage{amsmath, amssymb, amsthm}
\usepackage{mathtools}

\newtheorem{theorem}{Theorem}[section]
\newtheorem{proposition}[theorem]{Proposition}

\newtheorem{definition}[theorem]{Definition}
\newtheorem{remark}[theorem]{Remark}

\usepackage{graphicx}
\usepackage{booktabs}
\usepackage{caption}
\usepackage{subcaption}
\usepackage{float}
\usepackage{algpseudocode}
\usepackage{algorithm}
\usepackage{float}      
\usepackage{placeins}   
\usepackage{tikz}
\usetikzlibrary{positioning, arrows.meta, calc}

\usepackage[round]{natbib}
\usepackage[colorlinks=true, linkcolor=blue,
            citecolor=blue, urlcolor=blue]{hyperref}

\usepackage{xcolor}
\usepackage{enumitem}
\usepackage{microtype}

\title{ 
  \Large Bridging Stochastic Control and Deep Hedging: Structural Priors for No-Transaction Band Networks
}
\author{
  Jules Arzel\textsuperscript{1,2}, Noureddine Lehdili\textsuperscript{2}\\[6pt]
  \small \textsuperscript{1} Université Paris-Dauphine PSL, 
  \small \textsuperscript{2} Natixis CIB
}
\date{March 2026}

\begin{document}

\maketitle

\begin{abstract}
This paper studies the problem of hedging and pricing a European call
option under proportional transaction costs, from two complementary
perspectives. We first derive the optimal
hedging strategy under CARA utility, following the stochastic control
framework of \citet{DavisPanasZariphopoulou1993}, characterising the
no-transaction band via the Hamilton--Jacobi--Bellman Quasi-Variational
Inequality (HJBQVI) and the Whalley--Wilmott asymptotic approximation.
We then adopt a deep hedging approach, proposing two architectures that
build on the No-Transaction Band Network of \citet{ImakiImajoItoMinamiNakagawa2021}:
NTBN-$\Delta$, which makes delta-centring explicit, and WW-NTBN, which incorporates the
Whalley--Wilmott formula as a structural prior on the bandwidth and
replaces the hard clamp with a differentiable soft clamp. Numerical
experiments show that WW-NTBN converges faster, matches the stochastic
control no-transaction bands more closely, and generalises well across
transaction cost regimes. We further apply both frameworks to the bull
call spread, documenting the breakdown of price linearity under
transaction costs.
\end{abstract}


\input{introduction}        
\input{stochastic_control}  
\input{callspread}          
\input{deep_hedging}        
\input{results}             
\input{conclusion}          

\clearpage
\begin{center}
  {\large\bfseries Appendices}
\end{center}
\vspace{1em}

\titleformat{\section}[block]
  {\normalfont\normalsize\bfseries}
  {\appendixname~\thesection.}{0.5em}{}

\titleformat{\subsection}[block]
  {\normalfont\small\bfseries}
  {\thesubsection}{0.5em}{}

\appendix
\input{appendix}            

\newpage
\nocite{*}
\bibliographystyle{plainnat}
\bibliography{references}

\end{document}

%% file: Introduction.tex
\section{Introduction}
\label{sec:introduction}

\subsection{Motivation}

The Black-Scholes model \citep{BlackScholes1973} provides the standard
framework for option pricing and hedging. Its central result, that a
continuously rebalanced portfolio in the underlying asset perfectly
replicates any European payoff, underpins much of modern derivatives
theory. Yet continuous rebalancing is a mathematical idealisation. Every
transaction in the underlying asset incurs a cost: bid-ask spreads,
exchange fees, market impact, etc. Even modest proportional costs,
applied to the rapid rebalancing that perfect replication demands,
accumulate into charges that erode the value of any hedging strategy.
The assumptions of Black-Scholes are therefore not just imprecise but
qualitatively misleading in a market where transaction costs are present.

The consequences go further than a simple reduction in profitability.
Transaction costs fundamentally change the structure of optimal
strategies. Under frictions, continuous delta-hedging is no longer
optimal; it is dominated by strategies that trade only when the portfolio
drifts sufficiently far from the target position. This inactivity is not
a failure of the hedger but a rational response: the cost of rebalancing
must be weighed against the utility loss from holding a suboptimal
position. The resulting no-transaction band, a region of the state space
where it is optimal to do nothing, is the defining structural feature of
optimal hedging under proportional transaction costs and one of the
central objects of study in this paper.

From a theoretical perspective, the no-transaction band characterises
the geometry of the optimal control problem and connects the hedging
strategy to indifference pricing in a precise way. From a practical
perspective, it determines how often and how aggressively a hedger
should trade, directly controlling the running cost of maintaining a
derivatives book.

\subsection{Related Work}

The problem of option replication under proportional transaction costs
was first studied systematically by \citet{Leland1985}, who proposed
adjusting the Black-Scholes volatility upward to absorb the expected cost
of discrete rebalancing. While widely used in practice, this approach is
heuristic: it does not arise from a utility-maximisation argument and
does not characterise the optimal trading policy.

A rigorous treatment using stochastic control was initiated by
\citet{HodgesNeuberger1989}, who introduced the utility indifference
pricing approach and established the existence of a no-transaction region.
This framework was subsequently formalised by
\citet{DavisPanasZariphopoulou1993}, who formulated the problem as a
singular stochastic control problem, derived the associated
Hamilton-Jacobi-Bellman quasi-variational inequality (HJBQVI), and
provided the first analytical characterisation of the no-transaction band
boundaries. Related work by \citet{BoyleVorst1992} and
\citet{EdirisingheNaikUppal1993} addressed discrete-time replication with
transaction costs from a different angle, deriving bounds on replication
costs without a utility framework. The connection between transaction costs
and indifference pricing is surveyed in \citet{HendersonHobson2004}.

A key simplification was achieved by \citet{WhalleyWilmott1997}, who
derived an explicit asymptotic formula for the width of the
no-transaction band in the limit of small transaction costs. Their result
expresses the half-bandwidth as a function of the Black-Scholes gamma,
the risk aversion and the transaction cost level, giving an analytically
tractable approximation that remains accurate for small to moderate
costs. Related asymptotic results for small transaction costs were
obtained by \citet{KallsenMuhleKarbe2015} in a more general
semimartingale setting. This formula plays a central role in our work.

The deep hedging framework of
\citet{BuehlerGononTeichmannWood2019} recast option hedging as a
learning problem. A neural network is trained end-to-end to minimise a
convex risk measure over simulated market trajectories, without requiring
an explicit model for the price dynamics. This approach handles
transaction costs, path-dependent payoffs, and portfolio-level
constraints within a single unified framework, and has been shown to
produce competitive hedging strategies across a range of market
conditions. Extensions to rough volatility dynamics were studied by
\citet{HorvathTeichmannZuric2021}.

Within this framework, \citet{ImakiImajoItoMinamiNakagawa2021} proposed
the No-Transaction Band Network (NTBN), which outputs the lower and
upper boundaries of a no-transaction band at each time step rather than
a direct hedge position. The hedge is updated only when the previous
position falls outside this band, explicitly encoding the structural
feature known from stochastic control theory. This design improves
training stability and performance in high-friction regimes. In their
experimental implementation, \citeauthor{ImakiImajoItoMinamiNakagawa2021}
centre the band around the Black-Scholes delta by adding and subtracting
the network outputs from $\Delta_{\mathrm{BS}}$, though they do not
present this as an explicit architectural choice or provide theoretical
motivation for it. Beyond delta-centring, significant structural
information available from the Whalley--Wilmott approximation, in
particular the correct scaling of the bandwidth with transaction cost and
risk aversion, is not incorporated into the architecture.

\subsection{Contributions}

This paper makes the following contributions.

\paragraph{Unified theoretical exposition.}
We provide a self-contained derivation of the optimal hedging problem
under proportional transaction costs, following the stochastic control
framework of \citet{DavisPanasZariphopoulou1993} specialised to CARA
utility. We derive the HJBQVI, characterise the no-transaction band, and
present the Whalley--Wilmott asymptotic expansion with a sketch of the
key derivation steps in the appendix. While the individual results are
established in the literature, we assemble them into a unified treatment
that makes the connection to the deep hedging architectures below
explicit. We also treat the bull call spread as a case study illustrating
the breakdown of price linearity under transaction costs, a consequence
of the no-transaction band structure that is often overlooked in
practice.

\paragraph{NTBN-$\Delta$: making delta-centring explicit.}
The original NTBN of \citet{ImakiImajoItoMinamiNakagawa2021} centres its band around $\Delta_{\mathrm{BS}}$ in practice, but treats
this as an unmarked implementation detail rather than an architectural
principle. We formalise this choice as an explicit design decision,
motivated directly by the stochastic control analysis: the frictionless
optimal holding \eqref{eq:opt_hold_option} shows that the natural centre of the no-transaction
band is the Black-Scholes delta. The resulting architecture, NTBN-$\Delta$,
makes the delta-centring transparent, provides a clear theoretical
justification for it, and serves as the principled baseline on which
WW-NTBN is built.

\paragraph{WW-NTBN: a Whalley--Wilmott guided architecture.}
Building on NTBN-$\Delta$, we propose the Whalley--Wilmott
No-Transaction Band Network (WW-NTBN), which incorporates two further
improvements. First, the initial half-bandwidth is set to the
Whalley--Wilmott analytical approximation
\begin{equation}
  h_{\mathrm{WW}} = \left(\frac{3\lambda S\,
  \Gamma^2_{\mathrm{BS}}}{2\gamma}\right)^{1/3},
\end{equation}
and the network learns only a residual correction on top of this prior.
This gives the network the correct scaling with transaction cost
$\lambda$ and risk aversion $\gamma$ at initialisation, accelerating
convergence and improving generalisation across parameter regimes.
Second, the hard clamp used to enforce the band is replaced by a
differentiable soft clamp
\begin{equation}
  \mathrm{softclamp}(x, \ell, u)
  = \frac{\ell + u}{2}
  + \frac{u - \ell}{2}\,\tanh\!\left(\beta\,
  \frac{x - \frac{\ell+u}{2}}{\frac{u-\ell}{2}}\right),
\end{equation}
which provides non-zero gradients even when the previous hedge lies
outside the band, removing a source of vanishing gradients present in
the standard NTBN. We note that \citet{ImakiImajoItoMinamiNakagawa2021}
themselves flag the vanishing gradient issue in the clamped region as a
training tip, but do not propose a differentiable remedy; the soft clamp
addresses this directly.

We note that WW-NTBN partly sacrifices the model-agnostic character of
deep hedging. By incorporating the Black-Scholes delta and gamma as
structural inputs, the architecture implicitly assumes log-normal
dynamics for the underlying asset. This is a deliberate trade-off: in
exchange for this additional model assumption, WW-NTBN achieves faster
convergence, lower entropic risk at convergence, and no-transaction bands
that quantitatively agree with the stochastic control solution across a
range of transaction cost levels. In settings where the log-normal
assumption is inappropriate, NTBN-$\Delta$ remains the better choice.

\paragraph{Comprehensive numerical comparison.}
We conduct a systematic comparison of the Black-Scholes benchmark, the
stochastic control dynamic programming solution, the deep hedging MLP
baseline, NTBN-$\Delta$, and WW-NTBN under identical parameters and
across a range of transaction cost levels. We compare writer and buyer
indifference prices, no-transaction band geometry, terminal P\&L
distributions, and trade metrics. We further apply all methods to the
bull call spread, documenting the divergence between joint and naive
hedging strategies and its implications for pricing.

\subsection{Organisation}

Section~\ref{sec:stochastic_control} develops the stochastic control
framework: we introduce the optimisation problem, derive the HJBQVI,
characterise the no-transaction band, and present the Whalley--Wilmott
asymptotic formula together with the numerical dynamic programming
scheme. Section~\ref{sec:callspread} treats the bull call spread as an
application, establishing the non-linearity of prices and hedging
strategies under transaction costs. Section~\ref{sec:deep_hedging}
introduces the deep hedging framework and describes the three network
architectures: the MLP baseline, NTBN-$\Delta$, and WW-NTBN.
Section~\ref{sec:results} presents the numerical comparison.
Section~\ref{sec:conclusion} concludes. The Appendix collects the
omitted proofs.

\bigskip
\noindent\textbf{Notation.}
Throughout, $S_t$ denotes the price of the risky asset at time $t$,
$W_t$ the investor's wealth, $y_t$ the number of shares held, and
$\lambda > 0$ the proportional transaction cost rate, assumed symmetric.
We write $\sigma$ for volatility, $\mu$ for drift, $r$ for the risk-free
rate, $K$ for the option strike, and $T$ for maturity. The risk-aversion
parameter of the CARA utility $u(x) = 1 - e^{-\gamma x}$ is $\gamma >
0$. We write $\Phi$ for the standard normal CDF, $\phi$ for its density,
and $\Delta_{\mathrm{BS}}$, $\Gamma_{\mathrm{BS}}$ for the Black-Scholes
delta and gamma of the option under consideration.

%% file: stochastic_control.tex

\section{Optimal Hedging via Stochastic Control}
\label{sec:stochastic_control}


\subsection{Market Model and the Frictionless Problem}
\label{subsec:frictionless}

\subsubsection{Price Dynamics and Portfolio Construction}

We work on a filtered probability space $(\Omega, \mathcal{F},
\{\mathcal{F}_t\}_{t \in [0,T]}, \mathbb{P})$ supporting a standard
Brownian motion $Z_t$. The market consists of two assets. The first is
a risky asset whose price $S_t$ follows a geometric Brownian motion,
\begin{equation}
  dS_t = \mu S_t \, dt + \sigma S_t \, dZ_t,
  \label{eq:gbm}
\end{equation}
where $\mu \in \mathbb{R}$ is the drift and $\sigma > 0$ the volatility.
The second is a risk-free bond with dynamics $dB_t = r B_t \, dt$, so
that one unit of bond at time $t$ is worth $e^{r(T-t)}$ at maturity.

The investor constructs a portfolio by holding $y_t$ shares of the risky
asset and allocating the remainder of wealth to the bond. We denote by
$W_t = y_tS_t + B_t$ the total portfolio wealth at time $t$, so the bond position has
value $W_t - y_t S_t$. We assume the portfolio is self-financing: no
external cash flows occur, and any purchase of shares is funded by
selling bonds and vice versa. The change in wealth over an infinitesimal
interval $dt$ is then
\begin{equation}
  dW_t = y_t \, dS_t + (W_t - y_t S_t) \frac{dB_t}{B_t}
       = y_t \, dS_t + r(W_t - y_t S_t) \, dt.
  \label{eq:self_financing_raw}
\end{equation}
Substituting \eqref{eq:gbm} and rearranging,
\begin{equation}
  dW_t = r W_t \, dt + y_t (\mu - r) S_t \, dt + y_t \sigma S_t \, dZ_t.
  \label{eq:wealth_dynamics}
\end{equation}
It is convenient to write $\pi_t = y_t S_t / W_t$ for the fraction of
wealth invested in the risky asset, giving
\begin{equation}
  dW_t = \bigl[r + \pi_t(\mu - r)\bigr] W_t \, dt
        + \pi_t \sigma W_t \, dZ_t.
  \label{eq:wealth_fraction}
\end{equation}
The quantity $\mu - r$ is the excess return of the risky asset over the
risk-free rate, often called the risk premium. When $\mu > r$, investing
in the risky asset generates a positive expected excess return at the
cost of introducing volatility. The investor's problem is to choose
$\pi_t$ to balance this trade-off optimally.

\subsubsection{The Optimal Investment Problem and the HJB Equation}

The investor has a utility function $u : \mathbb{R} \to \mathbb{R}$,
twice differentiable, strictly increasing, and strictly concave. The
goal is to maximise expected utility of terminal wealth,
\begin{equation}
  \sup_{\{\pi_t\}_{t \in [0,T]}} \mathbb{E}\bigl[u(W_T)\bigr].
  \label{eq:objective}
\end{equation}
The optimisation is over adapted processes $\pi_t$ satisfying suitable
integrability conditions. To solve \eqref{eq:objective} using dynamic
programming, we define the \emph{value function} as the maximum
achievable expected utility starting from wealth $W$ at time $t$,
\begin{equation}
  U(t, W) = \sup_{\{\pi_s\}_{s \in [t,T]}}
  \mathbb{E}_t\bigl[u(W_T)\bigr],
  \label{eq:value_function}
\end{equation}
where $\mathbb{E}_t[\,\cdot\,]$ denotes expectation conditional on
$\mathcal{F}_t$. The value function encodes the full solution to the
problem: once $U$ is known, the optimal control at any state $(t, W)$
can be recovered by maximising over $\pi$.

By the dynamic programming principle, $U$ satisfies the
\emph{Hamilton-Jacobi-Bellman equation}
\begin{equation}
  \frac{\partial U}{\partial t}
  + \max_{\pi} \,\mathcal{L}^\pi U \;=\; 0,
  \label{eq:hjb_operator}
\end{equation}
where $\mathcal{L}^\pi$ is the controlled infinitesimal generator of
$W_t$ under strategy $\pi$. Applying It\^{o}'s formula to $U(t, W_t)$
along the dynamics \eqref{eq:wealth_fraction} gives
\begin{equation}
  \mathcal{L}^\pi U \;=\;
  \bigl[r + \pi(\mu - r)\bigr] W \, U_W
  + \tfrac{1}{2} \pi^2 \sigma^2 W^2 \, U_{WW},
\end{equation}
so the HJB equation reads explicitly
\begin{equation}
  U_t + \max_{\pi} \left\{
    \bigl[r + \pi(\mu-r)\bigr] W \, U_W
    + \tfrac{1}{2} \pi^2 \sigma^2 W^2 \, U_{WW}
  \right\} = 0,
  \label{eq:hjb_frictionless}
\end{equation}
with terminal condition $U(T, W) = u(W)$. Since $U_{WW} < 0$ by
concavity, the maximand is strictly concave in $\pi$ and the
first-order condition yields the \emph{optimal Merton fraction}
\citep{Merton1969}
\begin{equation}
  \pi^* = -\frac{(\mu - r) \, U_W}{\sigma^2 W \, U_{WW}}.
  \label{eq:merton}
\end{equation}
The quantity $-U_W / (W U_{WW})$ is the reciprocal of the Arrow-Pratt
measure of relative risk aversion. Equation \eqref{eq:merton} thus
states that the optimal risky allocation is proportional to the risk
premium $\mu - r$ and inversely proportional to risk aversion and
variance, a result that holds for any utility function admitting a
smooth value function.

\subsubsection{Solution under CARA Utility}

For the CARA utility $u(W) = 1 - e^{-\gamma W}$, the HJB equation
\eqref{eq:hjb_frictionless} admits an explicit solution. We make the
ansatz
\begin{equation}
  U(t, W) = 1 - e^{-\gamma \varphi(t) W + \psi(t)},
\end{equation}
substitute into \eqref{eq:hjb_frictionless}, and match coefficients.
This yields the system $\varphi'(t) + r\varphi(t) = 0$ and
$\psi'(t) = (\mu-r)^2/(2\sigma^2)$, with boundary conditions
$\varphi(T) = 1$ and $\psi(T) = 0$. Solving gives
\begin{equation}
  \varphi(t) = e^{r(T-t)}, \qquad
  \psi(t) = \frac{(\mu-r)^2}{2\sigma^2}(t - T),
\end{equation}
so that the value function is
\begin{equation}
  U(t, W) = 1 - \exp\!\left(-\gamma e^{r(T-t)} W
  + \frac{(\mu-r)^2}{2\sigma^2}(t-T)\right).
  \label{eq:cara_value}
\end{equation}
The full derivation is given in Appendix~\ref{app:cara_derivation}.
Substituting \eqref{eq:cara_value} into \eqref{eq:merton}, the optimal
fraction and the corresponding optimal number of shares are
\begin{equation}
  \pi^* = \frac{\mu - r}{\sigma^2 \gamma W e^{r(T-t)}},
  \qquad
  y^* = \frac{\mu - r}{\sigma^2 \gamma S_t e^{r(T-t)}}.
  \label{eq:cara_merton}
\end{equation}
Note that $y^*$ is independent of wealth $W$: under CARA utility, the
optimal number of shares held does not depend on how rich the investor
is, only on the current stock price and time to maturity. This is a
distinctive feature of exponential utility that will prove crucial when
we introduce transaction costs.

\subsubsection{Adding a Short Option Position}

We now suppose the investor is also short one European call option with
strike $K$ and maturity $T$, with payoff $\varphi(S_T) = (S_T - K)^+$.
The option is sold at time $0$ for a premium $p$, which is added to the
initial wealth. At maturity, the investor must deliver the payoff, so
the effective terminal wealth becomes
\begin{equation}
  W_T^{\mathrm{option}} = W_T - (S_T - K)^+.
  \label{eq:terminal_wealth_option}
\end{equation}
The portfolio dynamics \eqref{eq:wealth_dynamics} are unchanged: the
option position does not require any cash flow before maturity, so the
self-financing condition still holds. What changes is the terminal
condition of the optimisation problem.

Because $S_T$ appears in the terminal wealth \eqref{eq:terminal_wealth_option},
the stock price $S_t$ now enters the value function as a state variable.
The value function becomes $U(t, S, W)$, and applying Ito's formula now
produces cross terms between the $S$ and $W$ dynamics. The HJB equation
becomes
\begin{equation}
  \frac{\partial U}{\partial t}
  + \mu S U_S + \frac{1}{2}\sigma^2 S^2 U_{SS}
  + \max_\pi \left\{
      \bigl[r + \pi(\mu-r)\bigr] W U_W
      + \frac{1}{2}\pi^2\sigma^2 W^2 U_{WW}
      + \pi \sigma^2 S W U_{SW}
    \right\} = 0,
  \label{eq:hjb_option}
\end{equation}
with terminal condition $U(T, S, W) = u(W - (S-K)^+)$.
The additional term $\pi \sigma^2 S W U_{SW}$ arises because the Brownian
motion driving $S$ and $W$ is the same: the investor's wealth and the
option's value are exposed to the same source of risk. The first-order
condition now gives
\begin{equation}
  \pi^* = -\frac{(\mu-r) U_W}{\sigma^2 W U_{WW}}
           - \frac{S \, U_{SW}}{W \, U_{WW}},
  \label{eq:pi_option}
\end{equation}
which differs from the pure Merton solution \eqref{eq:merton} by the
second term, which reflects the hedging demand induced by the option.

A standard verification argument shows that, under CARA utility, the
value function with the option is related to the one without by
\begin{equation}
  U^{\mathrm{option}}(t, S, W) = U(t, W - f(t, S)),
  \label{eq:bs_duality}
\end{equation}
where $f(t, S)$ is the Black-Scholes price of the call
\citep{DavisPanasZariphopoulou1993}. In other words, the effect of the
short option position is simply to reduce the effective wealth by the
option's Black-Scholes value. This is the pricing-hedging duality of the
frictionless case: the fair price of the option is exactly $f(t, S)$.

Using \eqref{eq:bs_duality}, we can compute $U_{SW}$ and $U_{WW}$
explicitly in terms of $U$ and the Black-Scholes sensitivities. The
key step is that $\partial\tilde{W}/\partial S = -\Delta_{\mathrm{BS}}$
where $\tilde{W} = W - f(t,S)$, which gives
$U_{SW} = -\Delta_{\mathrm{BS}} \cdot U_{WW}$. Substituting into
\eqref{eq:pi_option} and using the explicit CARA value function, the
optimal number of shares held is
\begin{equation}
  y^* = \Delta_{\mathrm{BS}}(t, S) + \frac{\mu - r}{\sigma^2 \gamma S e^{r(T-t)}},
  \label{eq:opt_hold_option}
\end{equation}
where $\Delta_{\mathrm{BS}} = \partial f / \partial S$ is the
Black-Scholes delta of the option. This decomposition has a clean
interpretation. The first term, $\Delta_{\mathrm{BS}}$, is the standard
delta hedge: it replicates the option's exposure to moves in $S$ and
would be the entire solution in a risk-neutral world. The second term,
$(\mu-r)/(\sigma^2 \gamma S e^{r(T-t)})$, is exactly the Merton
speculative position the investor would hold in the absence of any
option \eqref{eq:cara_merton}. The two components are fully additive:
under CARA utility, the speculative demand depends only on the stock
price and time to maturity, not on the investor's wealth or other
positions, so the option introduces a pure hedging demand on top of
which the investor adds the same speculation as without the option.

When risk aversion $\gamma \to \infty$, the speculative term vanishes
and the investor holds the pure delta hedge $y^* \to
\Delta_{\mathrm{BS}}$: an infinitely risk-averse investor does not
speculate and hedges the option perfectly. As $\gamma \to 0$ the
investor becomes risk-neutral and the speculative position dominates,
with the delta hedge becoming relatively small.

With the frictionless case fully characterised, we now introduce
transaction costs, which fundamentally change the structure of the
optimal strategy.


\subsection{Transaction Costs: Problem Formulation}
\label{subsec:tc_formulation}

We now introduce proportional transaction costs at rate $\lambda > 0$,
symmetric for buying and selling. In this section we decompose the
portfolio into two separate accounts: $y_t$ shares of the risky asset
and a cash balance $X_t$, so that total wealth is
\begin{equation}
  W_t = X_t + y_t S_t.
  \label{eq:total_wealth_tc}
\end{equation}
Unlike the frictionless case, trading is no longer costless and the
portfolio is not self-financing: every transaction reduces total wealth
by $\lambda S_t$ per unit traded.\\
Trading is governed by two non-negative adapted processes $l_t$ (rate
of purchase) and $m_t$ (rate of sale), giving
\begin{align}
  dy_t  &= l_t \, dt - m_t \, dt, \label{eq:dy} \\
  dX_t  &= r X_t \, dt - (1+\lambda) S_t l_t \, dt
           + (1-\lambda) S_t m_t \, dt. \label{eq:dX}
\end{align}
Each purchase of $l_t \, dt$ shares costs $(1+\lambda)S_t l_t \, dt$
from the cash account (the factor $1+\lambda$ reflects the transaction
cost), while each sale of $m_t \, dt$ shares generates
$(1-\lambda)S_t m_t \, dt$ in cash.

\begin{remark}
To verify consistency with the frictionless case, set $\lambda = 0$ in
\eqref{eq:dy}--\eqref{eq:dX}. Total wealth $W_t = X_t + y_t S_t$
then satisfies
\begin{align}
  dW_t &= dX_t + y_t \, dS_t + S_t \, dy_t \notag \\
       &= \bigl(rX_t \, dt - S_t \, dy_t\bigr)
         + y_t(\mu S_t \, dt + \sigma S_t \, dZ_t)
         + S_t \, dy_t \notag \\
       &= r(X_t + y_t S_t) \, dt
         + y_t S_t (\mu - r) \, dt
         + y_t \sigma S_t \, dZ_t \notag \\
       &= r W_t \, dt
         + \pi_t(\mu - r) W_t \, dt
         + \pi_t \sigma W_t \, dZ_t,
\end{align}
where $\pi_t = y_t S_t / W_t$. This recovers \eqref{eq:wealth_fraction}
exactly, confirming that the transaction cost formulation is a consistent
extension of the frictionless one.
\end{remark}

The investor is short a European option with payoff $\varphi(S_T)$. The
optimisation problem is
\begin{equation}
  \sup_{(l,m) \geq 0} \,
  \mathbb{E}\bigl[u\bigl(X_T + y_T S_T - \lambda S_T |y_T|
  - \varphi(S_T)\bigr)\bigr],
  \label{eq:opt_problem}
\end{equation}
where the liquidation term $-\lambda S_T |y_T|$ accounts for the cost
of unwinding the remaining position at maturity. The associated value
function is
\begin{equation}
  U(t, S, y, X) = \sup_{(l,m) \geq 0}
  \mathbb{E}_{t,S,y,X}\bigl[u\bigl(X_T + y_T S_T
  - \lambda S_T |y_T| - \varphi(S_T)\bigr)\bigr].
  \label{eq:value_function_tc}
\end{equation}


\subsection{HJB Quasi-Variational Inequality and the No-Transaction Band}
\label{subsec:hjb_tc}

Applying It\^{o}'s formula to $U(t, S_t, y_t, X_t)$ along the dynamics
\eqref{eq:dy}--\eqref{eq:dX}, and noting that $dy_t$ and the controlled
part of $dX_t$ are pure finite-variation processes, the only
second-order term is $(dS_t)^2 = \sigma^2 S^2 \, dt$. Taking
expectations, the drift of $U$ must vanish at the optimum by the
dynamic programming principle, giving
\begin{equation}
  U_t + \mu S U_S + \frac{1}{2}\sigma^2 S^2 U_{SS} + rX U_X
  + l\bigl(U_y - (1+\lambda) S \, U_X\bigr)
  + m\bigl(-U_y + (1-\lambda) S \, U_X\bigr) = 0.
  \label{eq:drift_full}
\end{equation}
The controls $l, m \geq 0$ appear linearly in \eqref{eq:drift_full}.
Since they are unconstrained in magnitude, the only way
\eqref{eq:drift_full} can hold without driving the drift to $\pm\infty$
is if each coefficient is non-positive whenever the corresponding
control is active. This forces the problem into a quasi-variational inequality
structure. Define
\begin{align}
  A &= U_y - (1+\lambda) S \, U_X, \label{eq:A} \\
  B &= -U_y + (1-\lambda) S \, U_X, \label{eq:B} \\
  C &= -U_t - \mu S U_S
       - \tfrac{1}{2}\sigma^2 S^2 U_{SS} - rX U_X. \label{eq:C}
\end{align}
The Hamilton-Jacobi-Bellman equation then takes the form of the
quasi-variational inequality
\begin{equation}
  \max\bigl\{A,\, B,\, C\bigr\} = 0,
  \label{eq:vi}
\end{equation}
with terminal condition
\begin{equation}
  U(T, S, y, X) = u\bigl(X + yS - \lambda S|y| - \varphi(S)\bigr).
  \label{eq:terminal}
\end{equation}

\begin{definition}[Trading regions]
\label{def:regions}
The state space $(S, y, X)$ is partitioned into three regions according
to which term in \eqref{eq:vi} is active.
\begin{itemize}
  \item \emph{No-transaction region} $\mathcal{NT}$: $C = 0$ and
    $A, B \leq 0$. No trading occurs and $U$ satisfies the PDE
    $U_t + \mu S U_S + \frac{1}{2}\sigma^2 S^2 U_{SS} + rX U_X = 0$.
    The condition $A, B \leq 0$ is equivalent to
    \begin{equation}
      (1-\lambda)S \;\leq\; \frac{U_y}{U_X} \;\leq\; (1+\lambda)S.
      \label{eq:nt_condition}
    \end{equation}
    \item \emph{Buy region} $\mathcal{B}$: $A = 0$ and $B, C \leq 0$.
      The marginal value of a share exceeds the ask price
      ($U_y/U_X > (1+\lambda)S$), so the investor purchases shares
      continuously, maintaining $U_y/U_X = (1+\lambda)S$,
      until re-entering $\mathcal{NT}$ from below.
    
    \item \emph{Sell region} $\mathcal{S}$: $B = 0$ and $A, C \leq 0$.
      The marginal value of a share falls below the bid price
      ($U_y/U_X < (1-\lambda)S$), so the investor sells shares
      continuously, maintaining $U_y/U_X = (1-\lambda)S$,
      until re-entering $\mathcal{NT}$ from above.
\end{itemize}
\end{definition}

The ratio $U_y / U_X$ is the marginal value of an additional share
expressed in units of cash. The no-transaction condition
\eqref{eq:nt_condition} states that trading is suboptimal precisely
when this marginal value lies within the bid-ask spread
$[(1-\lambda)S,\, (1+\lambda)S]$: the utility gain from rebalancing
is smaller than the cost it incurs. When the marginal value of a share
exceeds the ask price $(1+\lambda)S$, the investor holds too few shares
and buys; when it falls below the bid price $(1-\lambda)S$, the investor
holds too many and sells. Simultaneous buying and selling ($l > 0$ and
$m > 0$) is never optimal, as it would waste transaction costs without
improving the portfolio.


\subsection{CARA Utility and Dimension Reduction}
\label{subsec:cara}

For general utility functions, the value function $U$ depends on all
four variables $(t, S, y, X)$, making direct numerical solution costly.
The CARA utility $u(x) = 1 - e^{-\gamma x}$ admits a factorisation
that eliminates the cash dimension entirely.

\begin{proposition}[Dimension reduction]
\label{prop:dim_reduction}
Under CARA utility, the value function takes the form
\begin{equation}
  U(t, S, y, X) = 1 - e^{-\gamma X / \delta(t,T)} \, Q(t, y, S),
  \label{eq:factorisation}
\end{equation}
where $\delta(t, T) = e^{-r(T-t)}$ is the discount factor and
$Q(t, y, S)$ solves the reduced quasi-variational inequality
\begin{equation}
  \min\left\{
    Q_y + \frac{\gamma(1+\lambda)S}{\delta(t,T)} Q, \;
    -Q_y - \frac{\gamma(1-\lambda)S}{\delta(t,T)} Q, \;
    -Q_t - \mu S Q_S - \tfrac{1}{2}\sigma^2 S^2 Q_{SS}
  \right\} = 0,
  \label{eq:Q_vi}
\end{equation}
with terminal condition
$Q(T, y, S) = \exp\bigl(-\gamma(yS - \lambda S|y| - \varphi(S))\bigr)$.
\end{proposition}
\begin{proof}
See Appendix~\ref{app:cara_derivation}.
\end{proof}

The key consequence is that $Q$ does not depend on $X$. Once $Q$ is
computed by solving \eqref{eq:Q_vi} on the three-dimensional grid
$(t, y, S)$, the full value function $U$ is recovered from
\eqref{eq:factorisation} for any initial cash balance $X$ at negligible
additional cost.


\subsection{Indifference Pricing}
\label{subsec:indifference}

Utility indifference pricing asks: what cash amount $p$ makes the
investor exactly indifferent between selling the option and receiving $p$,
versus not selling? Formally, the writer's indifference price $p^w$
satisfies
\begin{equation}
  U^{\mathrm{with}}(t, S, y, X + p^w) = U^{\mathrm{without}}(t, S, y, X),
\end{equation}
where the superscripts refer to the value functions with and without
the option liability respectively. Using the factorisation
\eqref{eq:factorisation}, the cash dependence cancels and one obtains
the explicit formula
\begin{equation}
  p^w(t, y, S) = \frac{\delta(t,T)}{\gamma}
    \left[\log Q^{\mathrm{without}}(t, y, S)
    - \log Q^{\mathrm{with, writer}}(t, y, S)\right].
  \label{eq:writer_price}
\end{equation}
The buyer's indifference price $p^b$ is defined symmetrically by
\begin{equation}
  p^b(t, y, S) = \frac{\delta(t,T)}{\gamma}
    \left[\log Q^{\mathrm{with, buyer}}(t, y, S)
    - \log Q^{\mathrm{without}}(t, y, S)\right].
  \label{eq:buyer_price}
\end{equation}
In the presence of transaction costs, $p^w \geq p^b$ in general,
generating a bid-ask spread that widens with both $\lambda$ and $\gamma$.
This spread vanishes as $\lambda \to 0$, where both prices converge to
the Black-Scholes price.

\begin{remark}
The factorisation \eqref{eq:factorisation} and the resulting pricing
formulas \eqref{eq:writer_price}--\eqref{eq:buyer_price} hold because
CARA utility is translation-invariant: $u(x + c) = e^{-\gamma c} u(x)$
up to a constant. This is the key analytic advantage of exponential
utility in the presence of transaction costs.
\end{remark}


\subsection{The Whalley-Wilmott Asymptotic Band}
\label{subsec:ww}
 
The quasi-variational inequality \eqref{eq:Q_vi} must in general be solved
numerically. However, in the limit of small transaction costs
$\lambda \to 0$, \citet{WhalleyWilmott1997} show that the width of the
no-transaction band can be obtained analytically by a matched asymptotic
expansion. The key insight is that when $\lambda$ is small, the
no-transaction band is a narrow region of width $O(\lambda^{1/3})$
around the frictionless optimal holding $y^*(t,S)$. This fractional
power arises naturally from the structure of the problem: the
suboptimality cost of staying within the band is quadratic in the
deviation from $y^*$, while the transaction cost of trading is linear
in the trade size, and balancing these two contributions leads to a
cubic optimisation whose minimiser scales as $\lambda^{1/3}$. A
heuristic sketch of the derivation is given in
Appendix~\ref{app:ww_derivation}.
 
\paragraph{The asymptotic formula.}
To leading order in $\lambda$, the no-transaction band around the
frictionless optimal holding $y^*(t, S)$ has half-width
\begin{equation}
  h_{\mathrm{WW}}(t, S) =
  \left(\frac{3\lambda\delta(t,T) S\,\Gamma_{\mathrm{BS}}(t,S)^2}
             {2\gamma}\right)^{1/3},
  \label{eq:ww_band}
\end{equation}
where $\delta(t,T) = e^{-r(T-t)}$ is the discount factor and
$\Gamma_{\mathrm{BS}}(t, S) = \partial^2 f/\partial S^2$ is the
Black-Scholes gamma of the option in price space. The optimal policy is
therefore to do nothing when $y_t \in [y^* - h_{\mathrm{WW}},\,
y^* + h_{\mathrm{WW}}]$, and otherwise to trade to the nearest
boundary.
 
\paragraph{Interpretation.}
The formula \eqref{eq:ww_band} has clear intuitive content. The
half-width $h_{\mathrm{WW}}$ increases with $\lambda$ (wider band when
trading is more costly) and with $\Gamma_{\mathrm{BS}}$ (wider band
when the option delta changes rapidly with $S$, making rebalancing
errors more persistent and thus more costly to correct). It decreases
with $\gamma$ (tighter band for more risk-averse investors who tolerate
less deviation from the optimal position). The $\lambda^{1/3}$ scaling
is a well-known feature of singular control problems with proportional
costs \citep{WhalleyWilmott1997}.
 
\paragraph{Validity.}
The Whalley-Wilmott formula is an asymptotic result valid for small
$\lambda$. For large transaction costs it underestimates the true band
width, as confirmed by our numerical experiments in
Section~\ref{sec:results}. Nevertheless, it provides an accurate and
computationally free approximation in the small to moderate cost regime
($\lambda \lesssim 1\%$) that is most relevant in practice.
 

\subsection{Numerical Scheme}
\label{subsec:numerical_sc}

We solve the reduced quasi-variational inequality \eqref{eq:Q_vi} by backward
induction on a binomial tree for $S$ combined with a uniform grid for
$y$. The CARA dimension reduction of Proposition~\ref{prop:dim_reduction}
is essential here: because $Q$ does not depend on $X$, the scheme
operates on a three-dimensional grid $(t, y, S)$ rather than the full
four-dimensional state space, making the computation tractable.

\paragraph{Discretisation.}
We discretise time into $N$ steps of size $\Delta t = T/N$. The stock
price follows a binomial approximation with up and down factors
\begin{equation}
  u = e^{\sigma\sqrt{\Delta t}}, \qquad d = e^{-\sigma\sqrt{\Delta t}},
  \label{eq:binomial_ud}
\end{equation}
and real-world probability $p = (e^{\mu \Delta t} - d)/(u - d)$.
The holding $y$ is discretised on a uniform grid
$y_k = k \Delta y$ for $k \in \{-M, \ldots, M\}$, where $M=int(0.8 \lfloor N/2 \rfloor)$, and $\Delta y = \sigma \Delta t$.
We denote by $Q(t, S_j, y_k)$ the approximate value of $Q$ at node
$(t, S_j, y_k)$.

\paragraph{Backward recursion.}
Starting from the terminal condition, three actions are evaluated at
each node: do nothing, buy $\Delta y$ shares, or sell $\Delta y$ shares.
The exponential factors in the buy and sell updates arise from the CARA
factorisation \eqref{eq:factorisation}: a cash outflow of
$(1+\lambda)S\Delta y$ at time $t$ multiplies $Q$ by
$\exp(-\gamma(1+\lambda)S\Delta y / \delta(t,T))$, since $Q$ absorbs
the discounted trading costs. The optimal value is the minimum of the
three options, since $Q > 0$ and $U = 1 - e^{-\gamma X/\delta}Q$ is
increasing in $-Q$. The full procedure is given in
Algorithm~\ref{alg:dp}.

\begin{algorithm}[H]
{\small
\caption{Dynamic programming scheme for $Q(t, y, S)$}
\label{alg:dp}
\begin{algorithmic}[1]
\Require Grid parameters $N$, $M$, $\Delta y$,
         option payoff $\varphi$,
         model parameters $\mu$, $\sigma$, $r$, $\lambda$, $\gamma$
\State Compute binomial tree nodes $S_j$ and probability $p$
       via \eqref{eq:binomial_ud}
\State Initialise holding grid $y_k = k\Delta y$,
       $k = -M, \ldots, M$
\For{each terminal node $(S_j, y_k)$}
  \State $Q(T, S_j, y_k) \leftarrow
         \exp\!\bigl(-\gamma\bigl(y_k S_j
         - \lambda S_j |y_k| - \varphi(S_j)\bigr)\bigr)$
\EndFor
\For{$n = N-1, \ldots, 0$}
  \For{each node $(S_j, y_k)$ at time $t = n\Delta t$}
    \State $Q^{\mathrm{NT}} \leftarrow
           p\, Q(t{+}\Delta t,\, S_j u,\, y_k)
           + (1-p)\, Q(t{+}\Delta t,\, S_j d,\, y_k)$
    \State $Q^{\mathrm{Buy}} \leftarrow
           \exp\!\left(\dfrac{-\gamma(1+\lambda)S_j\,\Delta y}
           {\delta(t,T)}\right)
           Q(t,\, S_j,\, y_{k+1})$
    \State $Q^{\mathrm{Sell}} \leftarrow
           \exp\!\left(\dfrac{\gamma(1-\lambda)S_j\,\Delta y}
           {\delta(t,T)}\right)
           Q(t,\, S_j,\, y_{k-1})$
    \State $Q(t, S_j, y_k) \leftarrow
           \min\bigl\{Q^{\mathrm{NT}},\,
           Q^{\mathrm{Buy}},\,
           Q^{\mathrm{Sell}}\bigr\}$
    \State Record optimal action in $\mathcal{A}(t, S_j, y_k)
           \in \{\mathrm{NT}, \mathrm{Buy}, \mathrm{Sell}\}$
  \EndFor
\EndFor
\State Run scheme twice to obtain $Q^{\mathrm{with}}$ and
       $Q^{\mathrm{without}}$ (with and without $\varphi$)
\State $\hat{p}^w \leftarrow \dfrac{\delta(0,T)}{\gamma}
       \Bigl[\log Q^{\mathrm{without}}(0, S_0, y_0)
       - \log Q^{\mathrm{with}}(0, S_0, y_0)\Bigr]$
\State \Return writer price $\hat{p}^w$, action map $\mathcal{A}$
\end{algorithmic}
}
\end{algorithm}

\paragraph{No-transaction band extraction.}
The action map $\mathcal{A}$ recorded during the backward pass gives
the full partition of the $(S, y)$ space into the three regions of
Definition~\ref{def:regions}. For each stock price $S_j$ and time $t$,
the no-transaction band boundaries are extracted as
\begin{equation}
  \ell^{\mathrm{SC}}(t, S_j)
  = \min\{y_k : \mathcal{A}(t, S_j, y_k) = \mathrm{NT}\},
  \qquad
  u^{\mathrm{SC}}(t, S_j)
  = \max\{y_k : \mathcal{A}(t, S_j, y_k) = \mathrm{NT}\}.
  \label{eq:sc_band_extract}
\end{equation}
These boundaries are used in Section~\ref{sec:results} to compare the
stochastic control solution directly with the no-transaction bands
learned by NTBN-$\Delta$ and WW-NTBN.

\paragraph{Buyer's price.}
The buyer's indifference price is computed by modifying the terminal
condition to reflect the buyer's position: the buyer receives the
payoff $\varphi(S_T)$ at maturity, so the terminal condition becomes
$Q^{\mathrm{buyer}}(T, S_j, y_k) = \exp(-\gamma(y_k S_j - \lambda
S_j|y_k| + \varphi(S_j)))$. The buyer's price then follows from
\eqref{eq:buyer_price} analogously to the writer's price.

%% file: callspread.tex

\section{The Bull Call Spread under Transaction Costs}
\label{sec:callspread}

\subsection{Definition and the Frictionless Case}

A bull call spread consists of a long position in a European call with
strike $K_1$ and a short position in a European call with the same
maturity $T$ and a higher strike $K_2 > K_1$, written on the same
underlying. The terminal payoff is
\begin{equation}
  \varphi_{\mathrm{CS}}(S_T)
  = (S_T - K_1)^+ - (S_T - K_2)^+.
  \label{eq:cs_payoff}
\end{equation}
This strategy allows the investor to benefit from a moderate rise in the
underlying while paying a lower premium than for a single call, since
the short leg partially offsets the cost of the long leg.

In the frictionless case, the pricing problem is linear. The
Black-Scholes PDE is linear in the option value, so the price of any
combination of options is simply the corresponding combination of
individual prices. For the bull call spread this gives
\begin{equation}
  p_{\mathrm{CS}} = f(t, S; K_1) - f(t, S; K_2),
  \label{eq:cs_frictionless}
\end{equation}
where $f(t, S; K)$ denotes the Black-Scholes price of a call with
strike $K$. Equivalently, using the duality relation
\eqref{eq:bs_duality}, the value function of an investor short the
spread satisfies
\begin{equation}
  U^{\mathrm{CS}}(t, S, W)
  = U\bigl(t, W - f(t,S;K_1) + f(t,S;K_2)\bigr),
\end{equation}
and the optimal hedge is $\Delta_{\mathrm{BS}}(t,S;K_1) -
\Delta_{\mathrm{BS}}(t,S;K_2)$, the difference of the two individual
deltas. Linearity of both pricing and hedging is a direct consequence
of the linearity of the frictionless HJB equation.

\subsection{Breakdown of Linearity under Transaction Costs}

Once proportional transaction costs are introduced, the linearity
argument breaks down at a precise point. Under transaction costs,
pricing and hedging are governed by the quasi-variational inequality
\eqref{eq:vi}, which is nonlinear. The key object that determines the
structure of the optimal strategy is the no-transaction band, whose
half-width is given to leading order by the Whalley-Wilmott formula
\eqref{eq:ww_band}:
\begin{equation}
  h_{\mathrm{WW}}(t, S)
  = \left(\frac{3\lambda \delta(t,T)S\,\Gamma^2(t,S)}{2\gamma}\right)^{1/3},
\end{equation}
where $\Gamma(t,S)$ is the gamma of the option being hedged. The width
of the no-transaction band is therefore a nonlinear function of the
gamma, scaling as $\Gamma^{2/3}$.

Now consider hedging the two legs of the call spread separately. Each
leg has its own no-transaction band, with half-widths
$h_{\mathrm{WW}}(t,S;K_1)$ and $h_{\mathrm{WW}}(t,S;K_2)$
respectively. When the two strategies are combined, a rebalancing trade
is triggered whenever the joint holding exits either individual band.
The effective no-transaction band for the combined position is therefore
the intersection of the two individual bands, which is generically much
narrower than either band alone. Near the money, where the gammas of
the two calls are largest and most similar, this intersection can
nearly vanish, forcing almost continuous rebalancing and generating
transaction costs that approach those of delta-hedging without any
band structure at all.

The correct approach is to treat the call spread as a single instrument
and apply the QVI framework directly to the combined payoff
\eqref{eq:cs_payoff}. The gamma of the spread is
\begin{equation}
  \Gamma_{\mathrm{CS}}(t, S)
  = \Gamma_{\mathrm{BS}}(t, S; K_1) - \Gamma_{\mathrm{BS}}(t, S; K_2),
  \label{eq:cs_gamma}
\end{equation}
which is strictly smaller in magnitude than either individual gamma
when the two calls are close to the money, since the two gammas
partially cancel. By the Whalley-Wilmott formula, this smaller gamma
yields a wider no-transaction band for the jointly hedged spread than
the intersection of the two individual bands. The joint hedger therefore
trades less frequently, incurs lower transaction costs, and consequently
prices the spread at a lower indifference price than the naive approach.

This non-linearity has a direct consequence for pricing. Define the
naive spread price as the difference of the two individual indifference
prices,
\begin{equation}
  p_{\mathrm{naive}} = p^w(K_1) - p^b(K_2),
\end{equation}
where $p^w(K_1)$ is the writer's indifference price for the $K_1$ call
and $p^b(K_2)$ is the buyer's indifference price for the $K_2$ call.
The joint indifference price $p_{\mathrm{CS}}$ satisfies
\begin{equation}
  p_{\mathrm{CS}} \leq p_{\mathrm{naive}},
\end{equation}
with equality only when $\lambda = 0$. The gap $p_{\mathrm{naive}} -
p_{\mathrm{CS}}$ grows with the transaction cost level $\lambda$ and
reflects the hedging inefficiency of treating the two legs
independently. This is confirmed numerically in Section~\ref{sec:results}.

\begin{remark}
The argument above illustrates a general principle: under transaction
costs, the price of a portfolio of derivatives is not the sum of the
individual prices. The indifference pricing functional is subadditive
for a writer, meaning that hedging multiple liabilities jointly is
always at least as efficient as hedging them separately. This is
precisely because the no-transaction bands of the joint problem are
wider than the intersection of the individual bands, reducing the total
cost of hedging.
\end{remark}

\subsection{Numerical Treatment}

The joint hedging problem for the call spread fits directly into the
framework of Section~\ref{sec:stochastic_control}. It suffices to
replace the terminal condition \eqref{eq:terminal} with
\begin{equation}
  U(T, S, y, X)
  = u\bigl(X + yS - \lambda S|y| - \varphi_{\mathrm{CS}}(S)\bigr),
\end{equation}
where $\varphi_{\mathrm{CS}}$ is given by \eqref{eq:cs_payoff}, and
run the dynamic programming scheme of
Section~\ref{subsec:numerical_sc} unchanged. The indifference price
is then computed via \eqref{eq:writer_price} by comparing the value
function with and without the spread liability.

For the naive approach, the two legs are priced and hedged
independently using the same scheme, and the resulting strategies are
superimposed. The comparison of the two approaches in terms of
no-transaction band geometry, terminal P\&L distributions, trade
metrics, and prices is presented in Section~\ref{sec:results}.

%% file: deep_hedging.tex

\section{Deep Hedging}
\label{sec:deep_hedging}

\subsection{From Stochastic Control to Neural Networks}

The stochastic control framework developed in
Section~\ref{sec:stochastic_control} provides a rigorous and
interpretable solution to the hedging problem under transaction costs.
However, it faces fundamental limitations as a practical tool. The
dynamic programming scheme requires discretising a state space of
dimension four $(t, S, y, X)$, which becomes computationally prohibitive
as the grid is refined. The approach also requires strong modelling
assumptions: the derivation of the HJB equation, the dimension reduction
of Proposition~\ref{prop:dim_reduction}, and the Whalley-Wilmott
approximation all rely on the specific structure of the GBM dynamics,
and extending the framework to richer price processes typically destroys
the tractability that makes it appealing. Finally, the scheme scales
poorly to portfolios of multiple derivatives, where the state space
dimension grows with the number of instruments.

Deep hedging \citep{BuehlerGononTeichmannWood2019} addresses these
limitations by replacing the explicit solution of the dynamic
programming equations with a direct optimisation over the parameters
of a neural network. The network learns a hedging policy from simulated
market trajectories, and the computational cost scales with the number
of paths rather than the dimension of the state space, making the
approach naturally suited to complex payoffs and multi-instrument
portfolios. The MLP and NTBN-$\Delta$ architectures we consider are
fully model-agnostic: they require only the ability to simulate paths
of the underlying, with no assumption on the price dynamics. The
WW-NTBN uses the Black-Scholes gamma as a structural prior, which
introduces a dependence on GBM dynamics, but this prior could in
principle be replaced by a model-specific sensitivity under richer
dynamics, preserving the same architectural advantages. As we will
show, the deep hedging approach produces hedging strategies and
indifference prices that are quantitatively consistent with the
stochastic control solution in the GBM setting, providing a valuable
cross-validation of both approaches.

\subsection{Problem Formulation and Loss Function}

We retain the same economic objective as in
Section~\ref{sec:stochastic_control}: the investor seeks to minimise
risk when holding a short position in a European option with payoff
$\varphi(S_T)$, subject to proportional transaction costs at rate
$\lambda$. For simplicity we set $r = 0$ throughout this section, so
discounting plays no role.

Rather than maximising expected CARA utility directly, we minimise the
\emph{entropic risk measure}
\begin{equation}
  \rho_\gamma(X) = \frac{1}{\gamma} \log \mathbb{E}\bigl[e^{-\gamma X}\bigr],
  \label{eq:erm}
\end{equation}
where $X$ is the terminal P\&L of the hedging strategy and $\gamma > 0$
is the risk-aversion parameter. Minimising $\rho_\gamma(X)$ is
equivalent to maximising $1-\mathbb{E}[-e^{-\gamma X}]$, which is the
objective associated with the CARA utility $u(x) = 1 -e^{-\gamma x}$. The
entropic risk measure is preferred numerically because it is a convex risk measure \citep{FoellmerSchied2002}, yielding a
well-behaved optimisation landscape with stable gradients.

The terminal P\&L for a hedging strategy $\{y_t\}_{t \leq T}$ is
\begin{equation}
  X = -\varphi(S_T)
    + \sum_{t=0}^{T-\Delta t} y_t (S_{t+\Delta t} - S_t)
    - \sum_{t=0}^{T-\Delta t} \lambda S_t |y_t - y_{t-1}|,
  \label{eq:pnl}
\end{equation}
where the first term is the option liability, the second is the
cumulative gain from the delta position, and the third is the total
transaction cost incurred. The hedge is initialised at $y_{-1} = 0$.

\paragraph{From optimal control to neural network parametrisation.}
In the stochastic control formulation of
Section~\ref{sec:stochastic_control}, the investor directly controls
the trading rate $(l_t, m_t)$ at each time $t$, and the optimal policy
is a function of the full state $(t, S_t, y_t, X_t)$.. In the deep
hedging framework, we instead seek the optimal policy within the class
of functions representable by a neural network. Formally, let
$I_t \in \mathbb{R}^d$ denote the feature vector available to the
investor at time $t$ (for instance, log-moneyness, time to maturity,
and volatility). A neural network with parameters $\theta$ defines a
deterministic policy
\begin{equation}
  y_t = \pi_\theta(I_t, y_{t-1}),
  \label{eq:nn_policy}
\end{equation}
mapping the current market information and previous position to the new
hedge. The original optimisation over adapted processes $\{y_t\}$ is
thereby replaced by a finite-dimensional optimisation over $\theta$:
\begin{equation}
  \min_\theta \; \rho_\gamma\bigl(X^\theta\bigr)
  = \min_\theta \; \frac{1}{\gamma}
    \log \mathbb{E}\bigl[e^{-\gamma X^\theta}\bigr],
  \label{eq:dh_objective}
\end{equation}
where $X^\theta$ denotes the terminal P\&L \eqref{eq:pnl} under the
policy $\pi_\theta$. This is a standard stochastic optimisation problem
that can be solved by gradient-based methods: given a batch of simulated
paths, the loss $\rho_\gamma(X^\theta)$ is evaluated, its gradient with
respect to $\theta$ is computed by backpropagation through the sequence
of hedge decisions \eqref{eq:nn_policy}, and $\theta$ is updated using
the Adam optimiser \citep{KingmaBa2015}.

The expressiveness of the policy class depends on the network
architecture. A sufficiently deep and wide network can approximate any
measurable function of $I_t$ \citep{Hornik1991}, so in principle the
neural network policy can converge to the true optimal policy as the
network size grows. In practice, the choice of architecture matters
greatly: a poorly designed architecture may require many more parameters
and training samples to learn the same policy. The three architectures
we consider in Section~\ref{subsec:architectures} differ precisely in
how much structural knowledge of the optimal policy is embedded in the
network design.

\paragraph{Neural network structure.}
A feedforward neural network with $L$ hidden layers computes
\begin{equation}
  x^{(0)} = I_t, \qquad
  x^{(\ell+1)} = \sigma\!\left(W^{(\ell)} x^{(\ell)} + b^{(\ell)}\right),
  \quad \ell = 0, \ldots, L-1,
  \label{eq:nn_layers}
\end{equation}
where $W^{(\ell)}$ and $b^{(\ell)}$ are the weight matrix and bias
vector of layer $\ell$, and $\sigma$ is a nonlinear activation function
applied elementwise. The output layer is linear and produces a
scalar or vector depending on the architecture. The activation function
introduces the nonlinearity necessary to approximate complex
input-output mappings; without it, the entire network would collapse to
a single affine map regardless of depth. We use the ReLU activation
$\sigma(x) = \max(0, x)$ throughout, which preserves gradient flow
during training while being computationally efficient. The full
parameter set is $\theta = \{W^{(\ell)}, b^{(\ell)}\}_{\ell=0}^{L}$,
optimised end-to-end by minimising \eqref{eq:dh_objective}.

\subsection{Network Architectures}
\label{subsec:architectures}

At each rebalancing date $t$, the network receives a feature vector
$I_t \in \mathbb{R}^d$ and outputs a hedging decision. The feature
vector consists of the log-moneyness $\log(S_t / K)$, the time to
maturity $T - t$, and the volatility $\sigma$. All three architectures
share this input structure; they differ in how they map $I_t$ to the
hedge $y_t$.

\subsubsection{MLP Baseline}

The first architecture is a standard feedforward Multi-Layer Perceptron
(MLP). At each time step, the network takes as input the feature vector
$I_t$ augmented with the previous hedge $y_{t-1}$, and outputs the new
hedge directly:
\begin{equation}
  y_t = \mathrm{MLP}_\theta(I_t, y_{t-1}).
\end{equation}
The network consists of five hidden layers with 32 neurons each and ReLU
activations, followed by a linear output layer producing a scalar. The
inclusion of $y_{t-1}$ as an input allows the network to account for
the transaction cost of moving from the current position, but
introduces path-dependence into the optimisation landscape: the network
must simultaneously learn the target hedge level and the cost of
reaching it from the current position. This action-dependence can slow
convergence and introduce instability, particularly at high transaction
cost levels. Despite this limitation, the MLP serves as a useful
baseline and provides insight into the behaviour of unconstrained
learned strategies.

\subsubsection{NTBN-$\Delta$}

The No-Transaction Band Network of
\citet{ImakiImajoItoMinamiNakagawa2021} addresses the action-dependence
of the MLP by explicitly modelling the no-transaction band structure
known from stochastic control theory. Rather than outputting a hedge
directly, the network outputs the lower and upper boundaries $(\ell_t,
u_t)$ of a no-transaction band, and the hedge is updated via a clamping
operation:
\begin{equation}
  y_t = \mathrm{clamp}(y_{t-1}, \ell_t, u_t)
      = \begin{cases}
          \ell_t & \text{if } y_{t-1} < \ell_t, \\
          y_{t-1} & \text{if } y_{t-1} \in [\ell_t, u_t], \\
          u_t & \text{if } y_{t-1} > u_t.
        \end{cases}
  \label{eq:clamp}
\end{equation}
This design removes $y_{t-1}$ from the network input: the network
receives only the market features $I_t$ and learns when to rebalance
rather than how much to rebalance, substantially simplifying the
optimisation problem.

In their experimental implementation,
\citet{ImakiImajoItoMinamiNakagawa2021} already centre the band around
the Black-Scholes delta: the two network outputs are added and
subtracted from $\Delta_{\mathrm{BS}}$ to obtain $u_t$ and $\ell_t$
respectively. However, this choice is not presented as an explicit
architectural decision and no theoretical motivation is given for it.
We formalise it here as a principled design choice, justified directly
by the stochastic control analysis: the frictionless optimum
\eqref{eq:opt_hold_option} shows that the natural centre of the
no-transaction band is the Black-Scholes delta $\Delta_{\mathrm{BS}}(t,
S)$, and centring the band there from the outset removes the need for
the network to discover this structure from data.

The resulting architecture, which we call \emph{NTBN-$\Delta$} to make
the delta-centring explicit, outputs two half-width corrections
$(\delta_\ell, \delta_u) \in \mathbb{R}^2$, and the band boundaries are
\begin{equation}
  \ell_t = \Delta_{\mathrm{BS}}(t, S_t) - \mathrm{LeakyReLU}(\delta_\ell),
  \qquad
  u_t = \Delta_{\mathrm{BS}}(t, S_t) + \mathrm{LeakyReLU}(\delta_u),
  \label{eq:ntbn_delta_band}
\end{equation}
where LeakyReLU ensures the half-widths are non-negative while
maintaining gradient flow. The hedge is then computed via \eqref{eq:clamp}
without introducing any additional model assumptions beyond those
already implicit in the Black-Scholes delta, it serves as the baseline on which WW-NTBN is built.

\begin{figure}[H]
\centering

\begin{subfigure}[b]{0.42\textwidth}
\centering
\begin{tikzpicture}[
  every node/.style = {font=\small},
  circ/.style = {circle, draw, minimum size=0.85cm, inner sep=1pt, align=center},
  box/.style  = {rectangle, draw, minimum width=1.8cm, minimum height=2.6cm, align=center},
  ln/.style   = {thin},
  labl/.style = {font=\footnotesize, align=center}
]

\node[circ] (It)   at (0,  0.65) {$I_{t}$};
\node[circ] (yt1)  at (0, -0.65) {$y_{t\text{-}1}$};
\node[box]  (NN)   at (2.0, 0)   {NN};
\node[circ] (yt)   at (3.9, 0)   {$y_{t}$};

\draw[ln] (It.east)  -- ($(NN.west)+(0, 0.4)$);
\draw[ln] (yt1.east) -- ($(NN.west)+(0,-0.4)$);
\draw[ln] (NN.east)  -- (yt.west);

\end{tikzpicture}
\caption{Feed-forward network (MLP).}
\label{fig:mlp_diagram}
\end{subfigure}
\hfill
\begin{subfigure}[b]{0.55\textwidth}
\centering
\begin{tikzpicture}[
  every node/.style = {font=\small},
  circ/.style = {circle, draw, minimum size=0.85cm, inner sep=1pt, align=center},
  box/.style  = {rectangle, draw, minimum width=1.8cm, minimum height=2.6cm, align=center},
  sbox/.style = {rectangle, draw, minimum width=1.3cm, minimum height=1.0cm, align=center},
  ln/.style   = {thin},
  labl/.style = {font=\footnotesize, align=center}
]

\node[circ] (It)   at (0,  0.65) {$I_{t}$};
\node[circ] (dbs)  at (0, -0.65) {$\Delta_{\mathrm{BS}}$};
\node[box]  (NN)   at (2.0, 0)   {NN};
\node[circ] (dl)   at (3.8,  0.65) {$\delta_{\ell}$};
\node[circ] (du)   at (3.8, -0.65) {$\delta_{u}$};
\node[circ] (lt)   at (5.5,  0.75) {$\ell_{t}$};
\node[circ] (yt1)  at (5.5,  0)    {$y_{t\text{-}1}$};
\node[circ] (ut)   at (5.5, -0.75) {$u_{t}$};
\node[sbox] (CL)   at (7.1,  0)    {clamp};
\node[circ] (yt)   at (8.7,  0)    {$y_{t}$};

\draw[ln] (It.east)  -- ($(NN.west)+(0, 0.4)$);
\draw[ln] (dbs.east) -- ($(NN.west)+(0,-0.4)$);

\draw[ln] ($(NN.east)+(0, 0.4)$) -- (dl.west);
\draw[ln] ($(NN.east)+(0,-0.4)$) -- (du.west);

\draw[ln] (dl.east) -- (lt.west);
\draw[ln] (du.east) -- (ut.west);

\draw[ln] (lt.east)  -- ($(CL.west)+(0, 0.35)$);
\draw[ln] (yt1.east) -- (CL.west);
\draw[ln] (ut.east)  -- ($(CL.west)+(0,-0.35)$);

\draw[ln] (CL.east) -- (yt.west);

\end{tikzpicture}
\caption{NTBN-$\Delta$.}
\label{fig:ntbn_delta_diagram}
\end{subfigure}

\caption{Comparison of the MLP and NTBN-$\Delta$ architectures.
$\ell_t = \Delta_{\mathrm{BS}} - \mathrm{LeakyReLU}(\delta_\ell)$ and
$u_t = \Delta_{\mathrm{BS}} + \mathrm{LeakyReLU}(\delta_u)$.
}
\label{fig:architecture_comparison}
\end{figure}
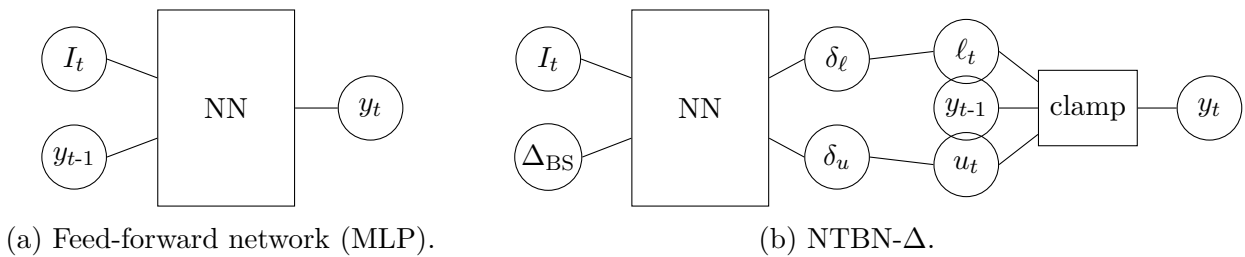

\subsubsection{WW-NTBN}

While NTBN-$\Delta$ correctly centres the band, it still leaves the
bandwidth to be learned entirely from data. The Whalley--Wilmott formula
\eqref{eq:ww_band} provides an analytical expression for the
half-bandwidth that is accurate for small to moderate transaction costs
and encodes the correct scaling with $\lambda$ and $\gamma$.
We propose to use this as a structural prior, allowing the network to
learn only a residual correction on top of the analytical baseline.

The \emph{Whalley--Wilmott No-Transaction Band Network} (WW-NTBN)
incorporates two improvements over NTBN-$\Delta$. The first is the
WW-guided bandwidth. Given the feature vector $I_t = (\log(S_t/K),
T-t, \sigma)$, we compute the Black-Scholes gamma
\begin{equation}
  \Gamma_{\mathrm{BS}}(t, S_t)
  = \frac{\phi(d_1)}{\sigma S_t \sqrt{T-t}},
  \qquad
  d_1 = \frac{\log(S_t/K) + \frac{1}{2}\sigma^2(T-t)}{\sigma\sqrt{T-t}},
  \label{eq:bs_gamma}
\end{equation}
and the Whalley--Wilmott half-width
\begin{equation}
  h_{\mathrm{WW}}(t, S_t)
  = \left(\frac{3\lambda S_t\,
    \Gamma^2_{\mathrm{BS}}(t, S_t)}{2\gamma}\right)^{1/3}.
  \label{eq:ww_prior}
\end{equation}
The MLP then outputs two residual corrections
$(\epsilon_\ell, \epsilon_u) \in \mathbb{R}^2$, and the band boundaries
are
\begin{equation}
  \ell_t = \Delta_{\mathrm{BS}}(t, S_t)
           - \mathrm{LeakyReLU}\bigl(h_{\mathrm{WW}}(t, S_t)
             + \epsilon_\ell\bigr),
  \qquad
  u_t = \Delta_{\mathrm{BS}}(t, S_t)
        + \mathrm{LeakyReLU}\bigl(h_{\mathrm{WW}}(t, S_t)
          + \epsilon_u\bigr).
  \label{eq:ww_ntbn_band}
\end{equation}
At initialisation, the residual corrections are zero and the band
reduces exactly to the Whalley--Wilmott approximation. The network
therefore begins training with the correct asymptotic structure and
learns corrections to it, rather than learning the band from scratch.

The second improvement is the replacement of the hard clamp
\eqref{eq:clamp} with a differentiable \emph{soft clamp}.
The hard
clamp has zero gradient with respect to $y_{t-1}$ whenever the previous
hedge lies outside the band, which removes the learning signal for
trajectories where the position is far from the band. We replace it
with
\begin{equation}
  \mathrm{softclamp}(x, \ell, u)
  = m + h \cdot \tanh\!\left(\beta \cdot \frac{x - m}{h}\right),
  \label{eq:softclamp}
\end{equation}
where $m = (\ell + u)/2$ is the band midpoint, $h = (u - \ell)/2$ is
the half-width, and $\beta > 0$ controls the sharpness of the
transition. For large $\beta$, \eqref{eq:softclamp} approximates the
hard clamp; for moderate $\beta$ (we use $\beta = 10$ in our
experiments), it provides non-zero gradients everywhere, giving the
network a learning signal even when the previous hedge is well outside
the band. The full forward pass of the WW-NTBN at time $t$ is
summarised in Algorithm~\ref{alg:ww_ntbn_forward}.

\begin{algorithm}[H]
\caption{WW-NTBN forward pass at time $t$}
\label{alg:ww_ntbn_forward}
\begin{algorithmic}[1]
\Require Features $I_t = (\log(S_t/K),\, T-t,\, \sigma)$,
         previous hedge $y_{t-1}$,
         parameters $\lambda$, $\gamma$, $\beta$
\State Compute $\Delta_{\mathrm{BS}}(t, S_t)$ and
       $\Gamma_{\mathrm{BS}}(t, S_t)$ from $I_t$
       \hfill (closed-form BS formulas)
\State $h_{\mathrm{WW}} \leftarrow
       \bigl(3\lambda S_t \Gamma^2_{\mathrm{BS}} /
       (2\gamma)\bigr)^{1/3}$
       \hfill (WW half-width prior)
\State $(\epsilon_\ell, \epsilon_u) \leftarrow \mathrm{MLP}_\theta(I_t)$
       \hfill (residual corrections)
\State $\ell_t \leftarrow \Delta_{\mathrm{BS}}
       - \mathrm{LeakyReLU}(h_{\mathrm{WW}} + \epsilon_\ell)$
\State $u_t \leftarrow \Delta_{\mathrm{BS}}
       + \mathrm{LeakyReLU}(h_{\mathrm{WW}} + \epsilon_u)$
\State $y_t \leftarrow \mathrm{softclamp}(y_{t-1}, \ell_t, u_t)$
       \hfill (differentiable band enforcement)
\State \Return $y_t$
\end{algorithmic}
\end{algorithm}

We note that WW-NTBN partly sacrifices the model-agnostic character of
deep hedging. By incorporating $\Delta_{\mathrm{BS}}$ and
$\Gamma_{\mathrm{BS}}$ as structural inputs, the architecture
implicitly assumes log-normal dynamics for the underlying asset. This
is a deliberate trade-off: the WW prior encodes structure that would
otherwise require many training samples to discover, yielding faster
convergence and better generalisation across transaction cost regimes.
In settings where the log-normal assumption is inappropriate,
NTBN-$\Delta$ remains the better choice.

\begin{figure}[H]
\centering
\begin{tikzpicture}[
  every node/.style = {font=\small},
  circ/.style = {
    circle, draw, minimum size=0.85cm,
    inner sep=1pt, align=center
  },
  box/.style = {
    rectangle, draw, rounded corners=3pt,
    minimum width=1.3cm, minimum height=1.6cm,
    align=center
  },
  sbox/.style = {
    rectangle, draw, rounded corners=3pt,
    minimum width=1.5cm, minimum height=1.6cm,
    align=center
  },
  arr/.style  = {->, >=stealth, thin},
  labl/.style = {font=\footnotesize, align=center}
]


\node[circ] (It)   at (0,    0)     {$I_{t}$};
\node[box]  (NN)   at (1.9,  0)     {NN};
\node[circ] (epsl) at (3.6,  0.65)  {$\varepsilon_{\ell}$};
\node[circ] (epsu) at (3.6, -0.65)  {$\varepsilon_{u}$};
\node[box]  (WW)   at (5.4,  0)     {WW\\prior};
\node[circ] (lt)   at (7.2,  0.75)  {$\ell_{t}$};
\node[circ] (yt1)  at (7.2,  0)     {$y_{t\text{-}1}$};
\node[circ] (ut)   at (7.2, -0.75)  {$u_{t}$};
\node[sbox] (SC)   at (9.1,  0)     {soft\\clamp};
\node[circ] (yt)   at (10.9, 0)     {$y_{t}$};


\draw[->, >=stealth, thin] (It) -- (NN);

\draw[->, >=stealth, thin] (NN.east) -- ++(0.2,0) |- (epsl.west);
\draw[->, >=stealth, thin] (NN.east) -- ++(0.2,0) |- (epsu.west);

\draw[->, >=stealth, thin]
  (It.north) -- ++(0, 0.65)
              -- ++(5.4, 0)
              -- (WW.north);

\draw[->, >=stealth, thin]
  (epsl.east) -- ($(WW.west)+(0, 0.4)$);

\draw[->, >=stealth, thin]
  (epsu.east) -- ($(WW.west)+(0,-0.4)$);

\draw[->, >=stealth, thin]
  ($(WW.east)+(0, 0.4)$) -- (lt.west);

\draw[->, >=stealth, thin]
  ($(WW.east)+(0,-0.4)$) -- (ut.west);

\draw[->, >=stealth, thin]
  (lt.east) -- ($(SC.west)+(0, 0.4)$);

\draw[->, >=stealth, thin]
  (yt1.east) -- (SC.west);

\draw[->, >=stealth, thin]
  (ut.east) -- ($(SC.west)+(0,-0.4)$);

\draw[->, >=stealth, thin] (SC) -- (yt);


\node[labl] at (1.9,  -1.5) {residual MLP};
\node[labl] at (5.4,  -1.5) {WW bandwidth};
\node[labl] at (9.1,  -1.5) {differentiable\\clamp};

\end{tikzpicture}
\caption{Forward pass of the WW-NTBN.}
\label{fig:ww_ntbn_diagram}
\end{figure}
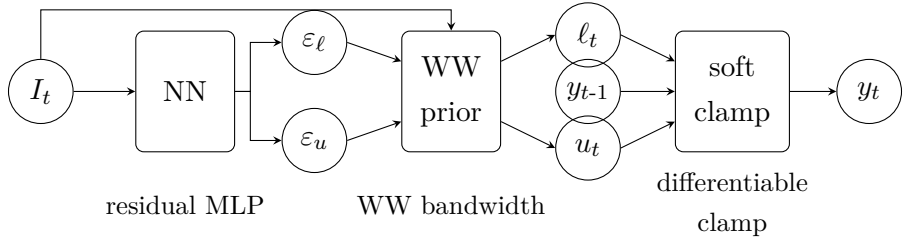

\subsection{Training Algorithm}

All three networks are trained using the same procedure, summarised in
Algorithm~\ref{alg:training}. The key steps are the simulation of GBM
paths, the application of the network policy to compute terminal P\&L,
and the minimisation of the entropic risk measure via stochastic
gradient descent.

\begin{algorithm}[H]
\caption{Deep hedging training}
\label{alg:training}
\begin{algorithmic}[1]
\Require Network parameters $\theta$, risk aversion $\gamma$,
         transaction cost $\lambda$, number of epochs $E$,
         batch size $N$, learning rate $\eta$
\For{$e = 1, \ldots, E$}
  \State Sample $N$ paths $\{S_t^{(i)}\}_{t=0}^T$ of GBM with
         parameters $(\mu, \sigma)$
  \For{each path $i = 1, \ldots, N$}
    \State Initialise $y_{-1}^{(i)} = 0$
    \For{$t = 0, \Delta t, \ldots, T - \Delta t$}
      \State Compute $I_t^{(i)}$ from $S_t^{(i)}$
      \State $y_t^{(i)} \leftarrow \pi_\theta(I_t^{(i)}, y_{t-1}^{(i)})$
             \hfill (network forward pass)
    \EndFor
    \State Compute terminal P\&L $X^{(i)}$ via \eqref{eq:pnl}
  \EndFor
  \State $\mathcal{L}(\theta) \leftarrow
         \frac{1}{\gamma}\log\!\left(\frac{1}{N}
         \sum_{i=1}^N e^{-\gamma X^{(i)}}\right)$
  \State $\theta \leftarrow \theta - \eta \nabla_\theta \mathcal{L}(\theta)$
         \hfill (Adam optimiser)
\EndFor
\State \Return $\theta$
\end{algorithmic}
\end{algorithm}

\subsection{Indifference Pricing}

Once the network is trained, the writer's indifference price is
estimated on an independent out-of-sample set of $N_{\mathrm{eval}}$
paths. For each path, two terminal P\&L values are computed: $X^{\mathrm{no}}$
using the trained strategy without the option liability, and
$X^{\mathrm{with}}$ with the liability included. The writer's price is
then
\begin{equation}
  \hat{p}^w = \frac{1}{\gamma}\log\!\left(
    \frac{\frac{1}{N}\sum_i e^{-\gamma X_i^{\mathrm{no}}}}
         {\frac{1}{N}\sum_i e^{-\gamma X_i^{\mathrm{with}}}}
  \right),
  \label{eq:dh_price}
\end{equation}
which is the empirical analogue of the theoretical formula
\eqref{eq:writer_price}. The buyer's price is computed symmetrically.
This direct correspondence between the deep hedging price
\eqref{eq:dh_price} and the stochastic control price
\eqref{eq:writer_price} is what makes the comparison in
Section~\ref{sec:results} meaningful: both approaches are computing
the same economic quantity, the utility indifference price, by
different methods.

%% file: results.tex

\section{Numerical Results}
\label{sec:results}

\subsection{Experimental Setup}
\label{subsec:setup}

All experiments use the following shared parameters: $S_0 = K = 1.0$,
$\sigma = 0.2$, $r = \mu = 0$, $\gamma = 1.0$, $T = 1.0$. Transaction
costs are tested at five levels: $\lambda \in \{0\%, 0.1\%, 0.5\%,
1\%, 5\%\}$. For the call spread, $K_1 = 0.9$ and $K_2 = 1.1$.

The stochastic control scheme (Algorithm~\ref{alg:dp}) uses $N = 400$
time steps and a holding grid of $2M + 1$ points. We note that the numerical scheme is restricted to trades of
exactly $\pm\Delta y$ shares at each step, which is a simplification
of the continuous-control problem and introduces discretisation error
in the optimal strategy. The deep hedging models are trained using the
Adam optimiser \citep{KingmaBa2015} with a learning rate of $10^{-2}$.
Following \citet{ImakiImajoItoMinamiNakagawa2021}, who showed that a
relatively high learning rate leads to faster and stable convergence
for band-based architectures, we use this value uniformly across all
three architectures to ensure a fair comparison. All networks consist
of five hidden layers with 32 units each and ReLU activations. Prices
are evaluated on $10{,}000$ out-of-sample paths via \eqref{eq:dh_price}.
The five models compared throughout are: the Black-Scholes closed-form
price (BS), the stochastic control dynamic programming solution (SC),
the MLP baseline, NTBN-$\Delta$, and WW-NTBN.

\FloatBarrier
\subsection{Frictionless Benchmark}
\label{subsec:tc0}

Before examining the effect of transaction costs, we validate all
models against the Black-Scholes closed-form price at $\lambda = 0$.
In the frictionless case with $\mu = r$, the indifference price under
CARA utility coincides with the Black-Scholes price regardless of
$\gamma$ (see Section~\ref{subsec:frictionless}), so all models should
recover $f_{\mathrm{BS}} = 0.07966$.

\begin{table}[H]
\centering
\caption{Option prices at $\lambda = 0$ compared to the
         Black-Scholes reference.}
\label{tab:tc0_prices}
\begin{tabular}{lccccc}
\toprule
Model  & BS      & SC      & MLP     & NTBN-$\Delta$ & WW-NTBN \\
\midrule
Price  & 0.07966 & 0.07995 & 0.07983 & 0.08060       & 0.07965 \\
Error  & --      & $+0.00029$ & $+0.00017$ & $+0.00094$ & $-0.00001$ \\
\bottomrule
\end{tabular}
\end{table}

Table~\ref{tab:tc0_prices} shows that all models are in close agreement
with the BS reference. WW-NTBN achieves the closest match with an error
of $-0.00001$, followed by MLP ($0.00017$) and SC ($0.00029$).
NTBN-$\Delta$ shows the largest deviation ($0.00094$), which we
attribute to the network requiring more training epochs to precisely
locate the band width when starting without a WW prior. These results
validate the implementations and confirm that all models correctly
reduce to the frictionless case when $\lambda = 0$.

\FloatBarrier
\subsection{Writer and Buyer Indifference Prices}
\label{subsec:prices}

\begin{table}[H]
\centering
\caption{Writer indifference prices for all models across transaction
         cost levels. The BS reference applies only at $\lambda = 0$.}
\label{tab:writer_prices}
\begin{tabular}{lccccc}
\toprule
$\lambda$ & BS      & SC      & MLP     & NTBN-$\Delta$ & WW-NTBN \\
\midrule
0\%       & 0.07966 & 0.07995 & 0.07983 & 0.08060       & 0.07965 \\
0.1\%     & --      & 0.08121 & 0.08174 & 0.08131       & 0.08095 \\
0.5\%     & --      & 0.08507 & 0.08426 & 0.08326       & 0.08338 \\
1\%       & --      & 0.08932 & 0.08639 & 0.08597       & 0.08591 \\
5\%       & --      & 0.11750 & 0.09005 & 0.08980       & 0.08878 \\
\bottomrule
\end{tabular}
\end{table}

Table~\ref{tab:writer_prices} reports writer indifference prices across
all transaction cost levels. Several patterns emerge.

First, all models agree that prices increase monotonically with
$\lambda$, consistent with the theoretical prediction that higher
transaction costs raise the minimum compensation required by the writer.

Second, the SC prices increase substantially faster with $\lambda$ than
the DH prices. At $\lambda = 5\%$, SC gives a writer price of $0.11750$
while WW-NTBN gives $0.08878$, a gap of nearly $0.03$. This divergence
has several contributing factors. On the SC side, the binomial tree
discretisation and the restriction to trades of exactly $\pm\Delta y$
per step mean the numerical scheme can only approximate the true optimal
continuous-control solution, and this approximation error is likely to
grow with $\lambda$ as the optimal strategy becomes more sensitive to
the exact band boundaries. On the DH side, the networks are trained by
minimising the empirical entropic risk over a finite batch of simulated
paths, which means extreme scenarios that are rare in the training
distribution may not be adequately priced. The SC solution, by
considering every node of the binomial tree including those arising from
large price movements, produces more conservative prices that account
for such scenarios. Finally, the two approaches use fundamentally
different time discretisations (a binomial tree for SC and Monte
Carlo paths for DH), which can produce systematically different price
estimates even when both approximate the same continuous-time problem.
We note this divergence as an interesting open question; separating
the contribution of each factor would require a more refined numerical
study.

Third, among the DH models, WW-NTBN consistently produces the lowest
writer prices, followed closely by NTBN-$\Delta$ and then MLP. The gap
between WW-NTBN and MLP grows with $\lambda$: at $\lambda = 5\%$ the
difference is $0.00127$, suggesting that the WW structural prior becomes
increasingly valuable at higher friction levels.

\begin{figure}[H]
  \centering
  \includegraphics[width=\textwidth]{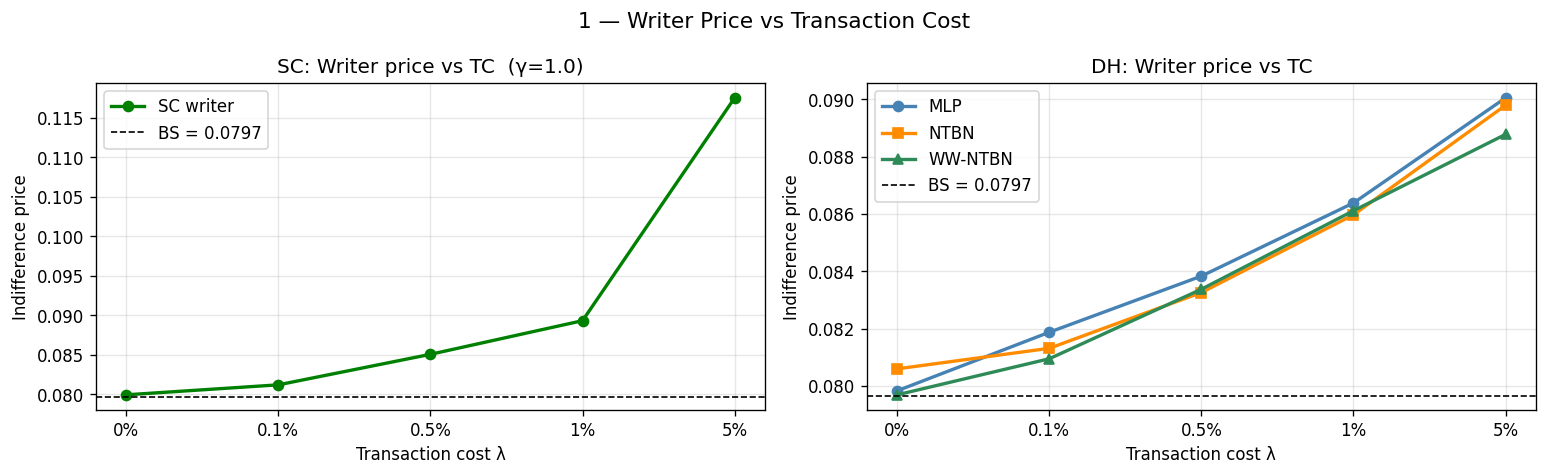}
  \caption{Writer indifference prices vs transaction cost level for
           all models. The BS reference price at $\lambda = 0$ is
           shown as a horizontal dashed line.}
  \label{fig:writer_prices}
\end{figure}

\begin{table}[H]
\centering
\caption{Writer and buyer indifference prices for SC and WW-NTBN,
         with resulting bid-ask spread.}
\label{tab:bidask}
\begin{tabular}{lcccccc}
\toprule
& \multicolumn{2}{c}{SC}
& \multicolumn{2}{c}{WW-NTBN}
& \multicolumn{2}{c}{Spread} \\
\cmidrule(lr){2-3}\cmidrule(lr){4-5}\cmidrule(lr){6-7}
$\lambda$ & Writer & Buyer & Writer & Buyer & SC     & WW-NTBN \\
\midrule
0\%       & 0.07995 & 0.07927 & 0.07969 & 0.07905 & 0.00068 & 0.00064 \\
0.1\%     & 0.08121 & 0.07802 & 0.08095 & 0.07696 & 0.00319 & 0.00399 \\
0.5\%     & 0.08507 & 0.07427 & 0.08338 & 0.07588 & 0.01080 & 0.00750 \\
1\%       & 0.08932 & 0.07037 & 0.08597 & 0.07272 & 0.01895 & 0.01325 \\
5\%       & 0.11750 & 0.04129 & 0.08878 & 0.07068 & 0.07621 & 0.01810 \\
\bottomrule
\end{tabular}
\end{table}

Table~\ref{tab:bidask} and Figure~\ref{fig:bidask} report writer and
buyer prices alongside the resulting bid-ask spread for SC and WW-NTBN.
A notable feature of the transaction cost framework is the appearance
of a bid-ask spread in indifference prices: even under the same model
and the same utility function, the writer demands more than the buyer
is willing to pay, because the two face different hedging problems.
The writer must hedge a short position, the buyer a long one, and
under proportional transaction costs these two hedging programmes have
different costs. Both approaches confirm the theoretical prediction of
Section~\ref{subsec:indifference}: the spread widens monotonically
with $\lambda$. The SC spread grows dramatically, reaching $0.07621$
at $\lambda = 5\%$, reflecting the conservative nature of the DP
solution under high frictions. The WW-NTBN spread grows more
moderately, reaching $0.01810$ at $\lambda = 5\%$, consistent with
the same factors discussed above for the writer prices.

\begin{figure}[H]
  \centering
  \includegraphics[width=\textwidth]{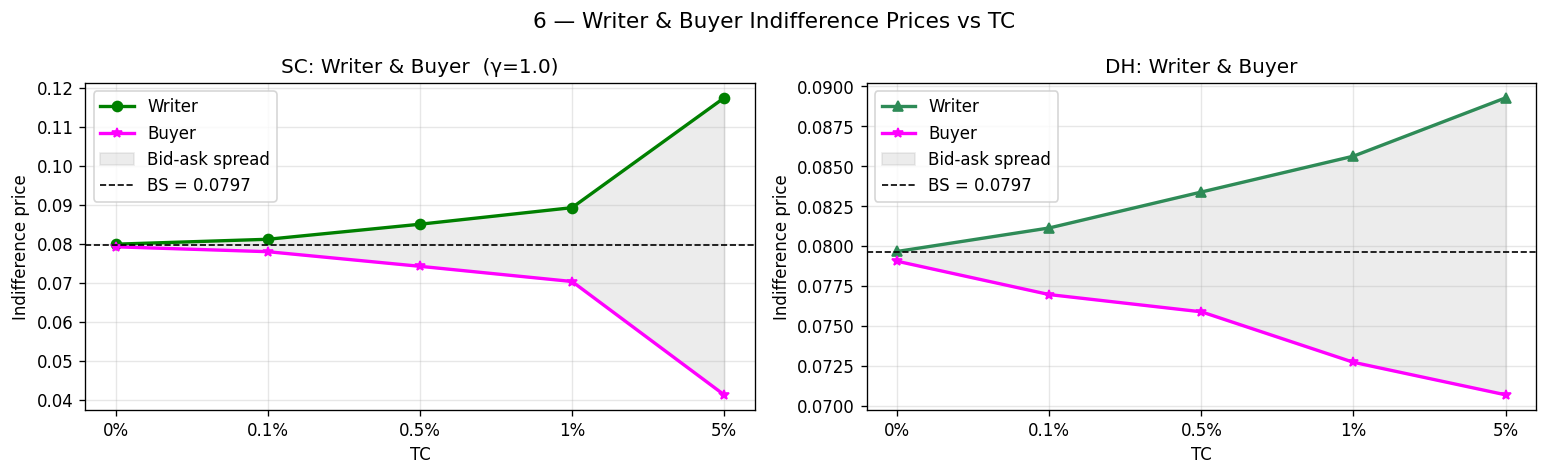}
  \caption{Writer and buyer indifference prices vs transaction cost
           level for SC (left) and WW-NTBN (right), illustrating the
           widening bid-ask spread.}
  \label{fig:bidask}
\end{figure}

\FloatBarrier
\subsection{No-Transaction Band Structure}
\label{subsec:bands}

Figure~\ref{fig:bands} compares the no-transaction band boundaries at
$t = 0.1$ extracted from the SC action map via \eqref{eq:sc_band_extract}
with those learned by NTBN-$\Delta$ and WW-NTBN, at $\lambda = 0.1\%$
and $\lambda = 1\%$ respectively. The bands are plotted as a function
of log-moneyness $x = \log(S/K)$, with the Black-Scholes delta shown
as a dashed reference.

At both TC levels, all models produce bands centred close to
$\Delta_{\mathrm{BS}}$, consistent with the stochastic control analysis.
WW-NTBN produces bands that match the SC boundaries more closely than
NTBN-$\Delta$ across the full range of log-moneyness, and this advantage
holds at both low and high TC. At low TC, the WW prior provides the
correct $\lambda^{1/3}$ bandwidth scaling from initialisation, so the
residual correction is small and the band rapidly converges to the SC
solution. At high TC, where the WW approximation is less exact, the
network learns a larger correction but still benefits from the structural
initialisation. NTBN-$\Delta$ must learn the bandwidth entirely from
data at all TC levels, requiring more training and introducing more
variance in the learned boundaries.

\begin{figure}[H]
  \centering
  \includegraphics[width=0.9\textwidth]{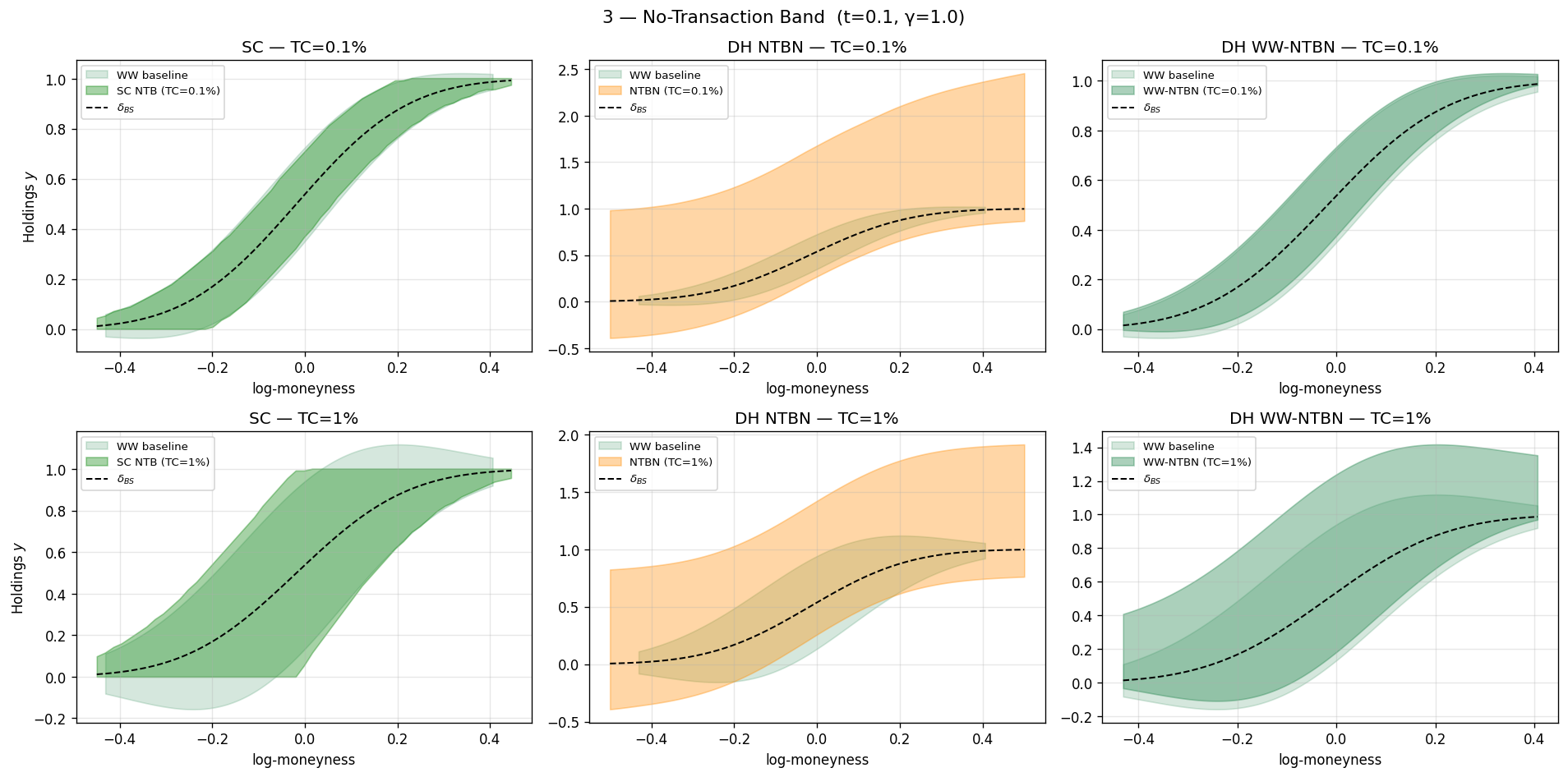}
  \caption{No-transaction band boundaries at $t = 0.1$ for
           $\lambda = 0.1\%$ and $\lambda = 1\%$.}
  \label{fig:bands}
\end{figure}

\FloatBarrier
\subsection{Training Convergence}
\label{subsec:convergence}

Figure~\ref{fig:loss} shows the entropic training loss as a function
of epoch for all three DH architectures at each TC level. WW-NTBN
converges to a lower final loss than both MLP and NTBN-$\Delta$ at all
TC levels and reaches this loss in fewer epochs. The advantage is most
pronounced at low TC ($\lambda \leq 0.5\%$), where the WW prior
provides an accurate initialisation.

The MLP consistently converges slowest and to the highest loss,
particularly at high TC where action-dependence introduces the most
instability. NTBN-$\Delta$ lies between the two: it benefits from the
delta-centring but must learn the bandwidth from scratch, requiring
more epochs than WW-NTBN but fewer than the MLP.

\begin{figure}[H]
  \centering
  \includegraphics[width=\textwidth]{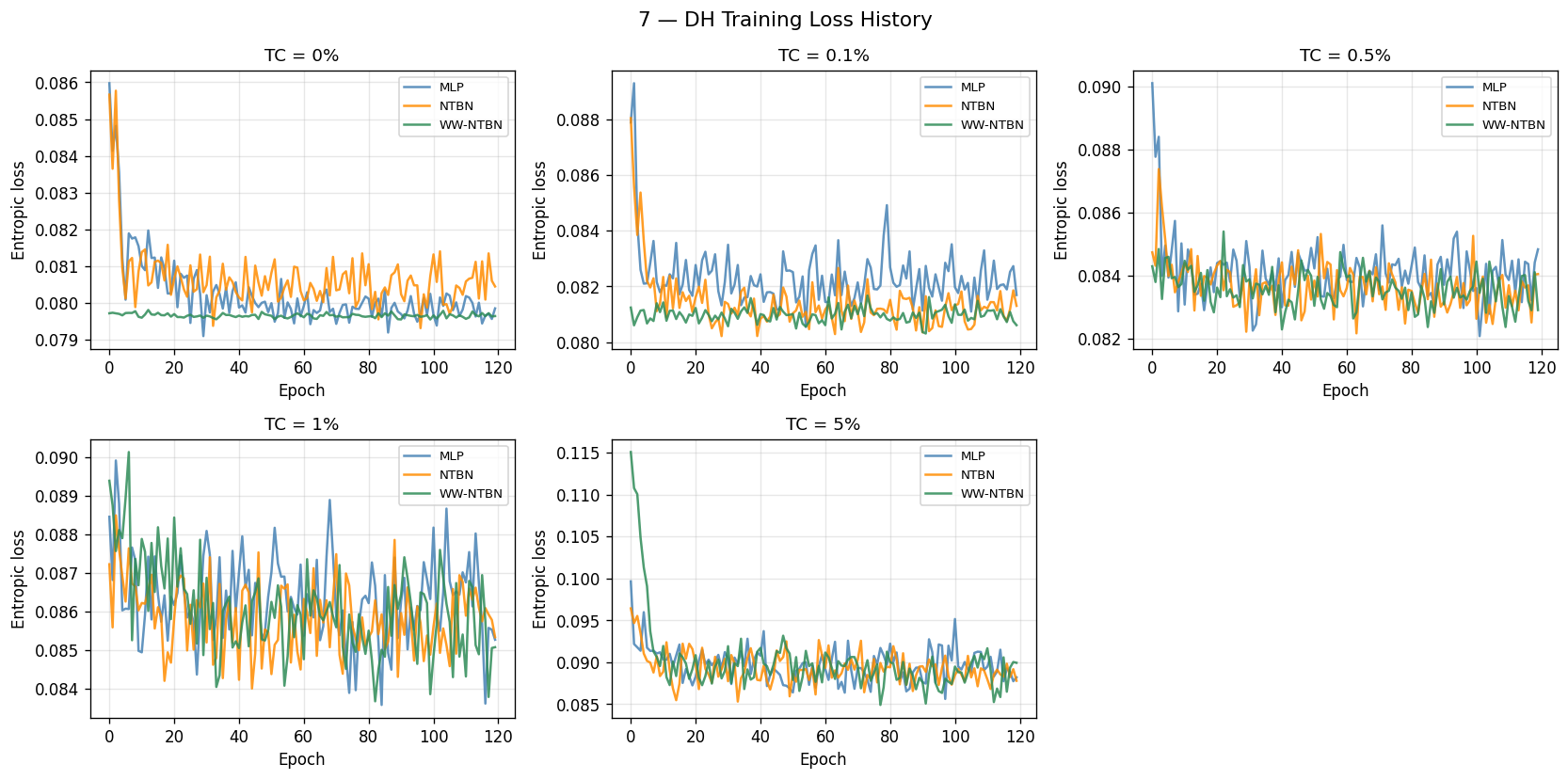}
  \caption{Entropic training loss vs epoch for MLP, NTBN-$\Delta$,
           and WW-NTBN at each of the five transaction cost levels.}
  \label{fig:loss}
\end{figure}

\FloatBarrier
\subsection{Terminal P\&L Distributions}
\label{subsec:pnl}

Figure~\ref{fig:pnl} shows the terminal P\&L distributions at
$\lambda = 0.5\%$ and $\lambda = 5\%$. All DH models share
approximately the same mean P\&L by construction of the indifference
pricing step. The SC solution achieves a slightly higher mean,
reflecting its conservative pricing.

The models differ more meaningfully in tail behaviour. SC exhibits a
worse CVaR than all three DH models, meaning its worst-case outcomes
are larger in magnitude. Among the DH models, WW-NTBN achieves the
best CVaR at both TC levels, followed by NTBN-$\Delta$ and then MLP.
This ordering is consistent with the band quality results of
Section~\ref{subsec:bands}: a more accurate no-transaction band reduces
unnecessary trading and limits the tail losses arising from excess
transaction costs.

\begin{figure}[H]
  \centering
  \includegraphics[width=0.68\textwidth]{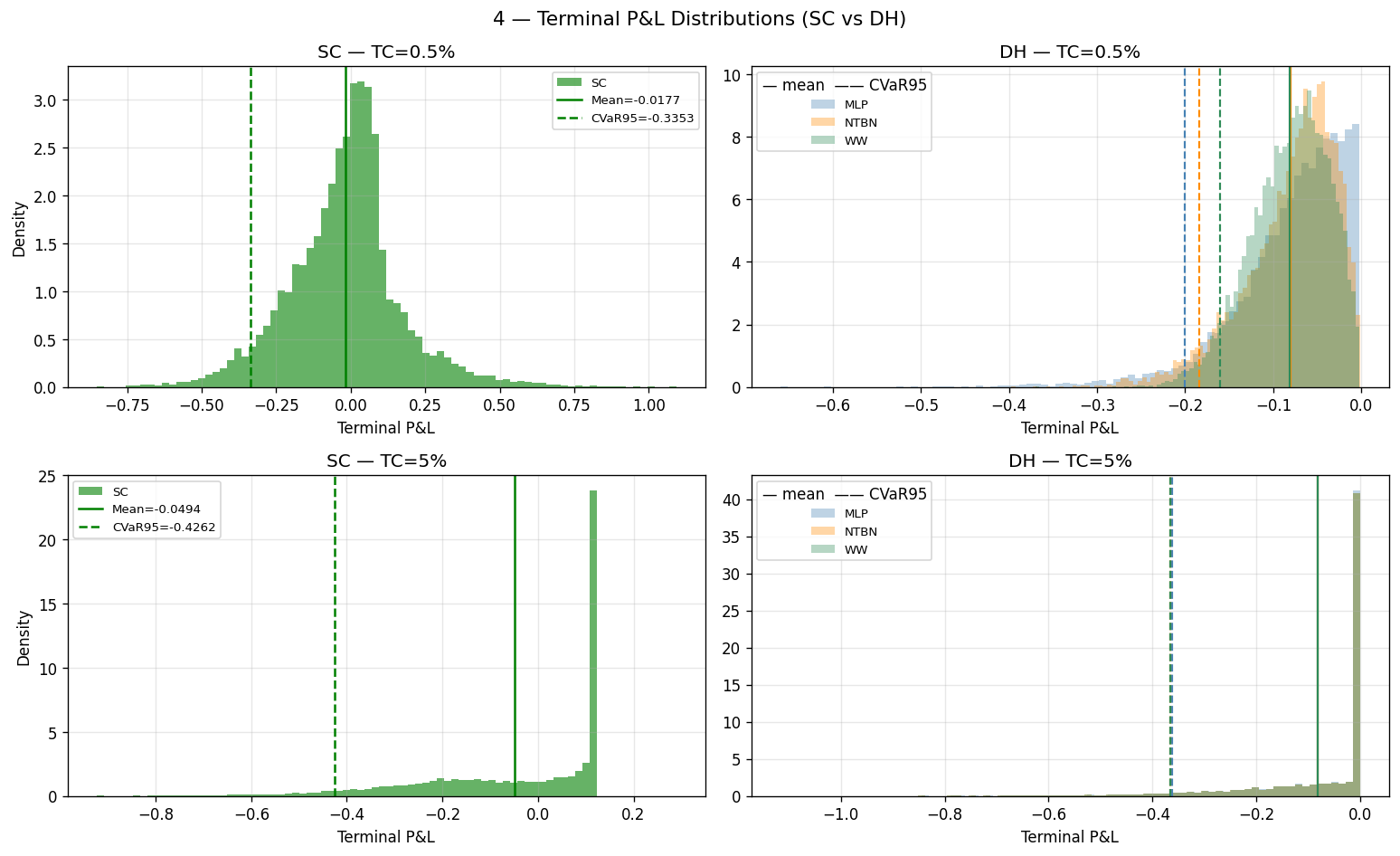}
  \caption{Terminal P\&L distributions for SC and DH models
           at $\lambda = 0.5\%$ and $\lambda = 5\%$ .}
  \label{fig:pnl}
\end{figure}

\FloatBarrier
\subsection{Trade Metrics}
\label{subsec:trademetrics}

To quantify the trading behaviour of each model, we introduce two
metrics. The \emph{trade frequency} (TF) measures the fraction of
nodes at which a rebalancing trade occurs,
\begin{equation}
  \mathrm{TF} = \frac{2}{(N+1)(N+2)}
  \sum_{t=1}^{N} \sum_{i=0}^{t-1}
  \sum_{j \in \{i, i+1\}}
  \mathbf{1}_{|y(t,j) - y(t-1,i)| > 0},
  \label{eq:tf}
\end{equation}
and the \emph{average shares traded} (AS) measures the average
absolute position change per node,
\begin{equation}
  \mathrm{AS} = \frac{2}{(N+1)(N+2)}
  \sum_{t=1}^{N} \sum_{i=0}^{t-1}
  \sum_{j \in \{i, i+1\}}
  |y(t,j) - y(t-1,i)|.
  \label{eq:as}
\end{equation}
Both metrics are computed over the binomial tree for SC and over
Monte Carlo paths for DH models.

Figure~\ref{fig:trademetrics} shows both metrics as a function of
$\lambda$. All models reduce trading frequency and position size
monotonically as $\lambda$ increases, consistent with the widening of
the no-transaction band. The band-based architectures (NTBN-$\Delta$
and WW-NTBN) trade significantly less frequently than the MLP at all
TC levels. The MLP must learn to reduce trading through the loss
function alone, which results in excess trading particularly at higher
TC levels.

The SC trade metrics should be interpreted with care. Because the
numerical scheme is restricted to trades of exactly $\pm\Delta y$
shares at each node, the SC strategy cannot execute the continuous
adjustments that the true optimal policy would make. This restriction
means the scheme naturally trades more frequently than a continuous
optimal policy would, and the average shares traded is mechanically
close to $\Delta y = 0.01$ per trade. Despite this limitation, the
qualitative behaviour is correct: both TF and AS decrease monotonically
with $\lambda$ for SC, confirming that the no-transaction band widens
as transaction costs increase, even if the absolute levels are
influenced by the discretisation.

\begin{figure}[H]
  \centering
  \includegraphics[width=\textwidth]{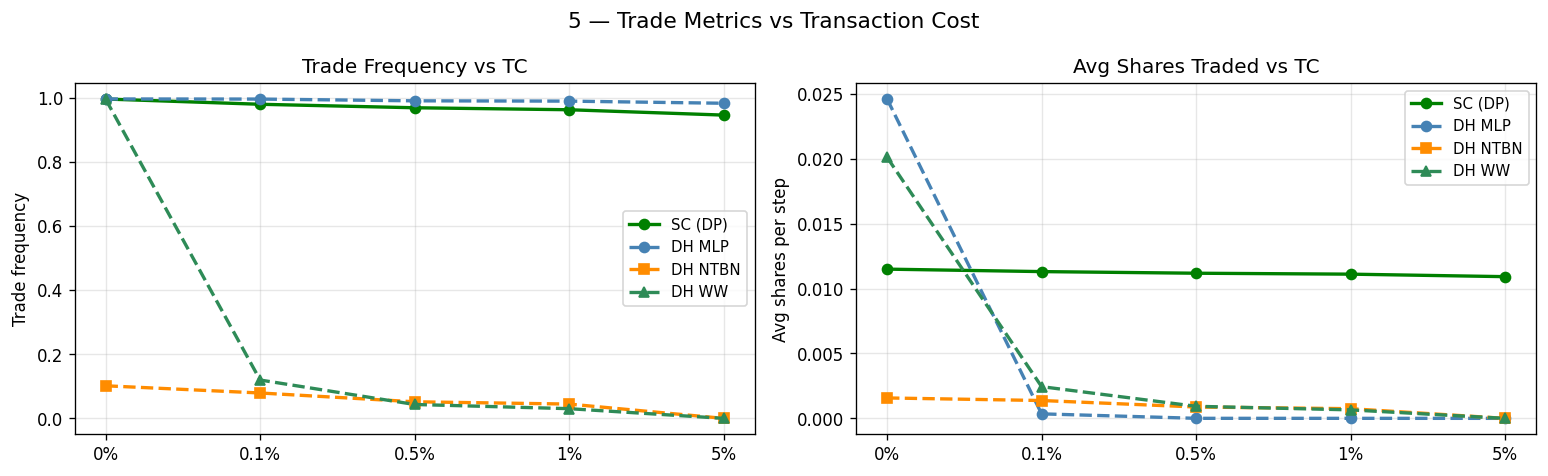}
  \caption{Trade frequency (left) and average shares traded (right)
           vs transaction cost level for all models.}
  \label{fig:trademetrics}
\end{figure}

\FloatBarrier
\subsection{Bull Call Spread: Pricing}
\label{subsec:cs_prices}

\begin{table}[H]
\centering
\caption{Bull call spread writer prices: joint vs naive hedging for
         SC and DH (WW-NTBN), with BS reference
         $p_{\mathrm{BS}} = 0.09324$.}
\label{tab:cs_prices}
\begin{tabular}{lcccc}
\toprule
$\lambda$ & SC Joint & SC Naive & DH Joint & DH Naive \\
\midrule
0\%       & 0.09324  & 0.09377  & 0.09352  & 0.09358  \\
0.1\%     & 0.09436  & 0.09400  & 0.09416  & 0.09745  \\
0.5\%     & 0.09503  & 0.09541  & 0.09524  & 0.10087  \\
1\%       & 0.09647  & 0.09733  & 0.09635  & 0.10564  \\
5\%       & 0.09677  & 0.11478  & 0.09677  & 0.11163  \\
\bottomrule
\end{tabular}
\end{table}

Table~\ref{tab:cs_prices} and Figure~\ref{fig:cs_prices} report joint
and naive writer prices for the bull call spread for both SC and DH.
At $\lambda = 0$, joint and naive prices are nearly identical for both
approaches, consistent with the linearity result of
Section~\ref{sec:callspread}. As $\lambda$ increases, the gap between
joint and naive prices widens substantially. At $\lambda = 5\%$, SC
naive reaches $0.11478$ while SC joint stays at $0.09677$, a difference
of $0.01801$; DH naive reaches $0.11163$ while DH joint is $0.09677$,
a difference of $0.01486$. This directly validates the subadditivity
argument of Section~\ref{sec:callspread}.

The SC joint price grows slowly with $\lambda$ and appears to approach
a ceiling around $0.097$. This reflects the structure of the spread
payoff: as the no-transaction band widens, the optimal strategy
requires increasingly less rebalancing, so the marginal hedging cost
of additional transaction costs diminishes. The DH joint price follows
the same pattern, confirming that both approaches capture the same
underlying phenomenon.

\begin{figure}[H]
  \centering
  \includegraphics[width=\textwidth]{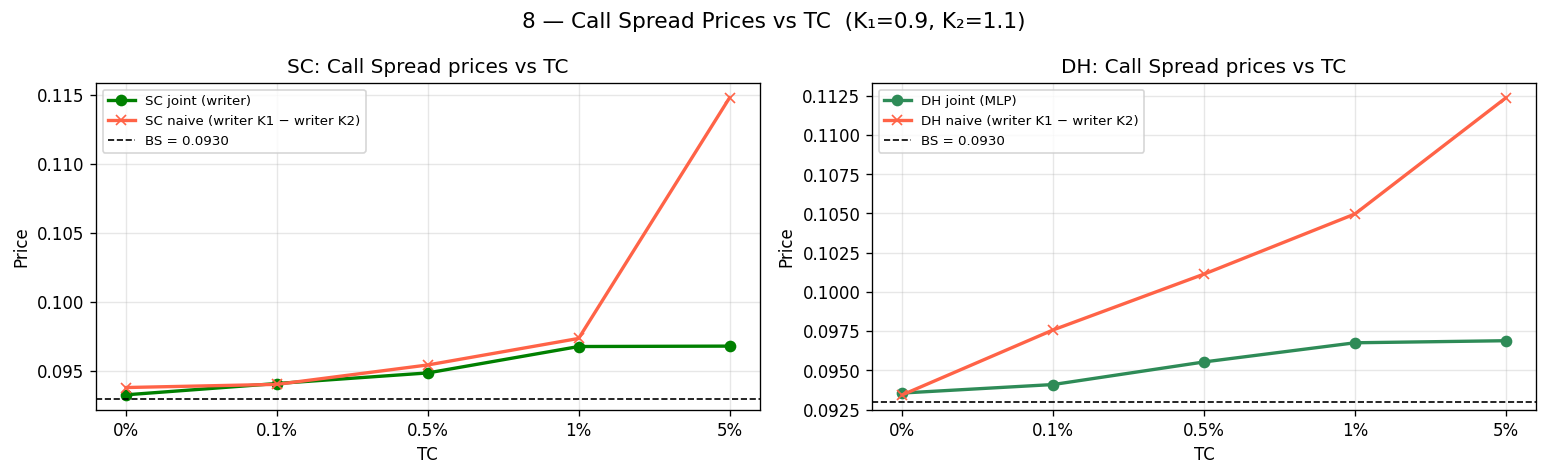}
  \caption{Bull call spread writer prices vs transaction cost level
           for joint and naive hedging strategies, under SC and DH.
           The BS reference price is shown as a dashed line.}
  \label{fig:cs_prices}
\end{figure}

For the remaining results in this section we focus on the stochastic
control solution, which provides the clearest structural illustration
of the phenomena of interest. The deep hedging results are qualitatively
consistent throughout.

\FloatBarrier
\subsection{Bull Call Spread: No-Transaction Band}
\label{subsec:cs_ntb}

Figure~\ref{fig:cs_band} compares the joint and naive no-transaction
bands at $t = 0.1$ for two TC levels. The joint band is centred around
the spread delta $\Delta_{\mathrm{BS}}(K_1) - \Delta_{\mathrm{BS}}(K_2)$
and widens substantially as $\lambda$ increases from $0.1\%$ to $1\%$,
consistent with the Whalley-Wilmott formula \eqref{eq:ww_band} and the
gamma cancellation argument of Section~\ref{sec:callspread}.

The naive band is the intersection of the writer band for the $K_1$
leg and the buyer band for the $K_2$ leg. This is the economically
correct decomposition: a trader hedging the two legs independently
would target $+\Delta_{\mathrm{BS}}(K_1)$ shares for the short $K_1$
call and $-\Delta_{\mathrm{BS}}(K_2)$ shares for the long $K_2$ call,
operating separate no-transaction bands around each target. The
intersection of these two bands in terms of the net holding $y$ is
the set of positions that simultaneously satisfies both individual
band constraints.

The figure reveals that this intersection is empty at both TC levels.
This is the direct consequence of the two individual bands being
centred at $+\Delta_{\mathrm{BS}}(K_1)$ and $-\Delta_{\mathrm{BS}}(K_2)$
respectively, which are on opposite sides of the spread delta near the
money. A naive hedger therefore has no region of inaction in the most
active part of the state space and is forced into near-continuous
rebalancing precisely where the spread's gamma exposure is largest.
The joint hedger, by contrast, maintains a wide band centred around
$\Delta_{\mathrm{BS}}(K_1) - \Delta_{\mathrm{BS}}(K_2)$ throughout,
trading only when the net holding drifts sufficiently far from the
spread delta. This contrast graphically explains both the higher trading
frequency and the higher indifference price of the naive approach.

\begin{figure}[H]
  \centering
  \includegraphics[width=\textwidth]{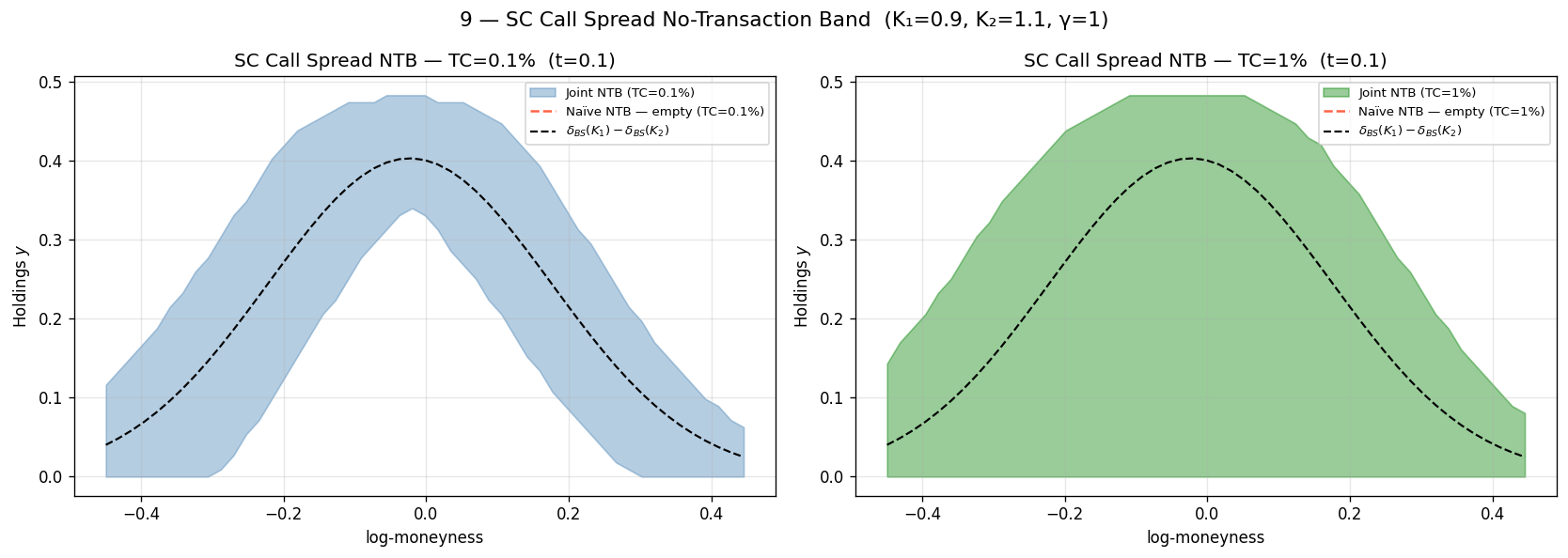}
  \caption{No-transaction bands for the jointly hedged bull call spread (blue/green) and the naive intersection band (red) at $t = 0.1$, for $\lambda = 0.1\%$ (left) and $\lambda = 1\%$ (right). The dashed line is the spread delta $\Delta_{\mathrm{BS}}(K_1) - \Delta_{\mathrm{BS}}(K_2)$.}
  \label{fig:cs_band}
\end{figure}

\FloatBarrier
\subsection{Bull Call Spread: P\&L and Trade Metrics}
\label{subsec:cs_trademetrics}

Figure~\ref{fig:cs_scatter} shows the terminal P\&L scatter at
$\lambda = 0.5\%$ for SC joint and naive hedging. We focus on this
TC level because it lies in the regime where transaction costs are
large enough to create a meaningful difference between the two
strategies, but below the level at which the joint policy switches to
near-inactivity, which would make the comparison less informative. The
joint strategy produces a tighter concentration of outcomes around the
mean, while the naive strategy exhibits wider dispersion and larger
outlier losses, consistent with the excess rebalancing that the empty
naive band forces near the money.

\begin{figure}[H]
  \centering
  \includegraphics[width=\textwidth]{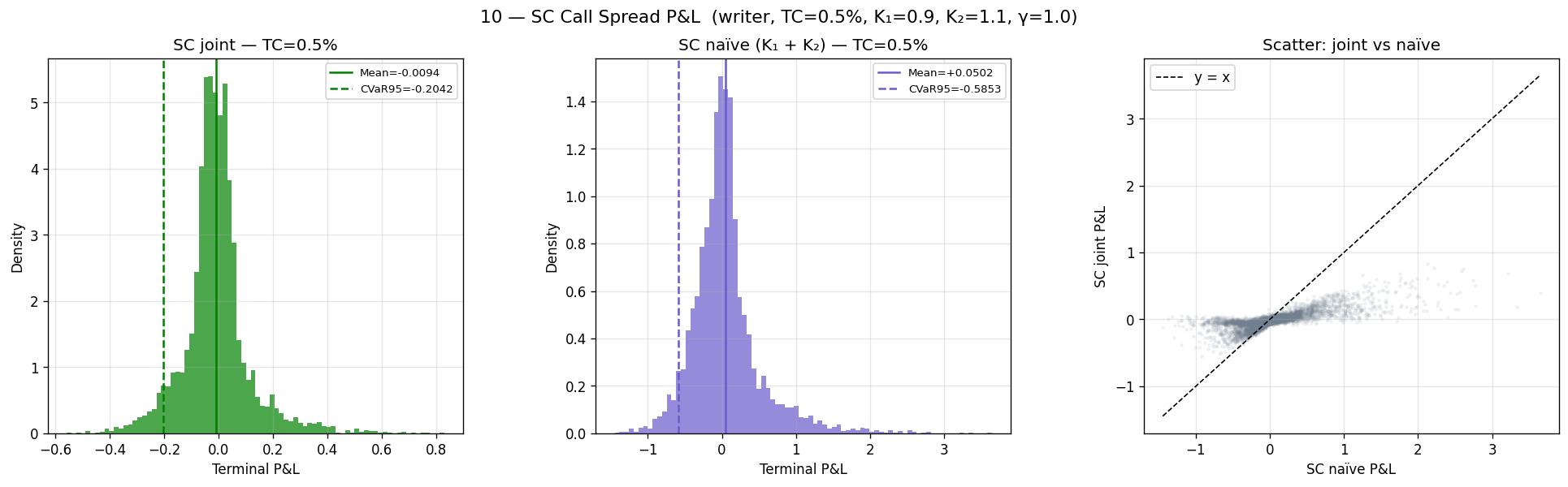}
  \caption{Terminal P\&L scatter at $\lambda = 0.5\%$ for SC joint
           and SC naive hedging of the bull call spread.}
  \label{fig:cs_scatter}
\end{figure}

Figure~\ref{fig:cs_trademetrics} confirms the mechanism quantitatively:
the joint strategy trades substantially less frequently than the naive
strategy at all TC levels, and the gap widens with $\lambda$. At low
TC levels the naive strategy trades nearly continuously, directly
reflecting the empty naive no-transaction band documented in
Section~\ref{subsec:cs_ntb}. This reduced trading activity under joint
hedging is the direct cause of the lower indifference prices documented
in Section~\ref{subsec:cs_prices}.

\begin{figure}[H]
  \centering
  \includegraphics[width=\textwidth]{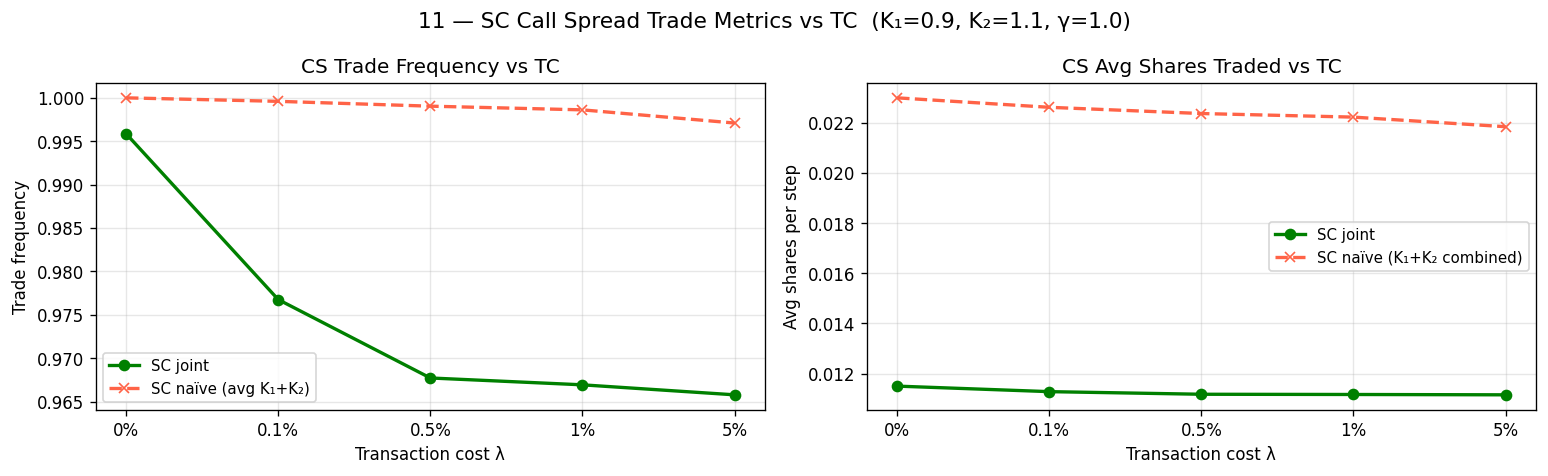}
  \caption{Trade frequency (left) and average shares traded (right)
           for joint and naive hedging of the bull call spread
           under SC, vs transaction cost level.}
  \label{fig:cs_trademetrics}
\end{figure}

\FloatBarrier
\subsection{Summary}
\label{subsec:summary}

\begin{table}[H]
\centering
\caption{Qualitative summary of model performance at $\lambda = 1\%$.}
\label{tab:summary}
\begin{tabular}{lcccc}
\toprule
Model          & Writer Price & Bid-Ask Spread & CVaR (95\%) & Band Quality \\
\midrule
SC             & 0.08932      & 0.01895        & Worst        & Reference    \\
MLP            & 0.08639      & --             & Moderate     & None         \\
NTBN-$\Delta$  & 0.08597      & --             & Good         & Moderate     \\
WW-NTBN        & 0.08591      & 0.01339        & Best         & Best         \\
\bottomrule
\end{tabular}
\end{table}

Table~\ref{tab:summary} provides a summary of model performance at
$\lambda = 1\%$. WW-NTBN achieves the best CVaR among DH models, the
tightest bid-ask spread, and the closest match to the SC
no-transaction band. SC provides the most conservative writer price
but the worst tail risk, reflecting its worst-case optimality
guarantees. The MLP, despite having no structural prior, produces
competitive prices but trades more and exhibits worse tail behaviour
than the band-based architectures. These results support the central
conclusion of this paper: embedding the structural knowledge from
stochastic control theory into the network architecture, as done in
WW-NTBN, improves both the quality of the learned hedging strategy and
its quantitative agreement with the analytically optimal solution.

\FloatBarrier

%% file: conclusion.tex

\section{Conclusion}
\label{sec:conclusion}

This paper has studied the problem of hedging and pricing European
options under proportional transaction costs from two complementary
perspectives, and shown that the two approaches are more deeply
connected than is typically appreciated in the literature.

On the stochastic control side, we derived the HJB quasi-variational
inequality governing the optimal strategy under CARA utility,
characterised the no-transaction band analytically, and presented the
Whalley-Wilmott asymptotic formula as a tractable approximation whose
accuracy we verified numerically. The CARA dimension reduction reduces
a four-dimensional PDE to a three-dimensional problem, making the
dynamic programming scheme computationally viable and enabling exact
indifference pricing via a single pair of backward recursions. The
bull call spread analysis showed that this framework immediately
reveals a phenomenon that is easy to miss in frictionless models: once
transaction costs are present, the price of a portfolio of derivatives
is strictly less than the sum of the individual prices when hedged
jointly, and the gap grows substantially with the transaction cost
level.

On the deep hedging side, the main contribution is the WW-NTBN
architecture, which uses the Whalley-Wilmott half-width as a structural
prior on the no-transaction bandwidth and replaces the hard clamp with
a differentiable soft clamp. By embedding the analytical result
directly into the network, WW-NTBN begins training with the correct
scaling in $\lambda$ and $\gamma$, converges faster than
both the MLP and NTBN-$\Delta$ baselines, and produces no-transaction
bands that quantitatively match the stochastic control solution across
all transaction cost levels tested. The intermediate architecture
NTBN-$\Delta$, which centres the band around the Black-Scholes delta
without imposing a prior on its width, provides a useful model-agnostic baseline when
the log-normal assumption is not appropriate.

The numerical comparison reveals a growing divergence between SC and
DH prices at high transaction costs that deserves comment. The gap
reflects several compounding factors: the SC numerical scheme is
restricted to trades of exactly $\pm\Delta y$ shares per step and uses
a binomial tree discretisation, both of which introduce approximation
errors that grow with $\lambda$; the DH networks are trained on
finite batches of simulated paths, which means rare extreme scenarios
may not be fully reflected in the learned policy; and the two approaches
use fundamentally different time discretisations, making direct
quantitative comparison imprecise. Neither approach is uniformly
superior. SC provides rigorous optimality guarantees and exact
indifference prices under its model assumptions, while deep hedging is
more flexible, scales to complex payoffs and portfolios, and does not
require solving a high-dimensional PDE. WW-NTBN occupies a deliberate
middle ground: it sacrifices some of the model-agnosticism of pure
deep hedging in exchange for the structural accuracy of stochastic
control, a trade-off that is most valuable when the transaction cost
level is moderate and the log-normal assumption is reasonable.

Several directions suggest themselves for future work. The most natural
extension is to apply WW-NTBN beyond the GBM setting, replacing the
Black-Scholes gamma with a model-specific hedging sensitivity under
stochastic volatility or jump-diffusion dynamics, where the qualitative
structure of the no-transaction band is preserved but the frictionless
optimum changes. A second direction is to extend the joint hedging
framework to larger portfolios of derivatives, where the subadditivity
of indifference pricing has potentially significant practical
implications for how dealers price and risk-manage structured products.
Finally, the WW prior could serve as a warm start for reinforcement
learning approaches to hedging, providing a structured initialisation
that reduces the exploration cost in the early stages of training.

%% file: appendix.tex

\appendix

\section{CARA Indirect Utility Derivation}
\label{app:cara_derivation}

\subsection{Without Option Position}

We verify the explicit form of the value function under CARA utility
$u(W) = 1 - e^{-\gamma W}$ in the frictionless case without an option.
We make the ansatz
\begin{equation}
  U(t, W) = 1 - e^{-\gamma \varphi(t) W + \psi(t)},
\end{equation}
and substitute into the HJB equation \eqref{eq:hjb_frictionless}. The
required partial derivatives are
\begin{align}
  U_t    &= \bigl(\gamma \varphi'(t) W - \psi'(t)\bigr)
             e^{-\gamma\varphi(t)W + \psi(t)}, \\
  U_W    &= \gamma\varphi(t)\, e^{-\gamma\varphi(t)W + \psi(t)}, \\
  U_{WW} &= -\gamma^2\varphi(t)^2\, e^{-\gamma\varphi(t)W + \psi(t)}.
\end{align}
Substituting the optimal Merton fraction \eqref{eq:merton} into the
HJB equation and dividing through by
$e^{-\gamma\varphi(t)W+\psi(t)} \neq 0$, we obtain
\begin{equation}
  \gamma\varphi'(t)W - \psi'(t)
  + \gamma r \varphi(t) W
  + \frac{(\mu-r)^2}{2\sigma^2} = 0.
\end{equation}
Matching coefficients in $W$ and in the constant term separately gives
the system
\begin{align}
  \varphi'(t) + r\varphi(t) &= 0, \label{eq:phi_ode} \\
  \psi'(t) &= \frac{(\mu-r)^2}{2\sigma^2}. \label{eq:psi_ode}
\end{align}
With boundary conditions $\varphi(T) = 1$ and $\psi(T) = 0$, the
solutions are
\begin{equation}
  \varphi(t) = e^{r(T-t)}, \qquad
  \psi(t) = \frac{(\mu-r)^2}{2\sigma^2}(t - T),
\end{equation}
giving the value function \eqref{eq:cara_value}. The optimal fraction
\eqref{eq:cara_merton} follows immediately from \eqref{eq:merton}.

\subsection{Dimension Reduction under Transaction Costs}

We prove Proposition~\ref{prop:dim_reduction}. The argument proceeds in
three steps: we first express $X_T$ explicitly as a function of $X_t$
and the controls, then use the CARA structure to factor out the $X_t$
dependence from the value function, and finally derive the reduced
quasi-variational inequality for $Q$.

\paragraph{\small Step 1: Solution of the cash dynamics.}
The cash balance satisfies $dX_t = rX_t\,dt - (1+\lambda)S_t l_t\,dt
+ (1-\lambda)S_t m_t\,dt$. Solving this linear SDE gives
\begin{equation}
  X_T = \frac{X_t}{\delta(t,T)}
      + \int_t^T \frac{1}{\delta(s,T)}
        \bigl[(1-\lambda)S_s\,dm_s - (1+\lambda)S_s\,dl_s\bigr],
  \label{eq:XT_explicit}
\end{equation}
where $\delta(t,T) = e^{-r(T-t)}$. The first term is the initial cash
balance compounded at the risk-free rate. The integral collects the
discounted net proceeds from all trading between $t$ and $T$. Crucially,
only the first term depends on $X_t$; the integral depends on the
controls $(l,m)$ and the stock price path, but not on the initial cash
balance.

\paragraph{\small Step 2: Factorisation of the terminal utility.}
Substituting \eqref{eq:XT_explicit} into the terminal utility
$u(X_T + y_T S_T - \lambda S_T|y_T| - \varphi(S_T))$ and writing
$\mathcal{G}^{l,m}$ for all terms in the argument that do not depend
on $X_t$,
\begin{align}
  &u\bigl(X_T + y_T S_T - \lambda S_T|y_T| - \varphi(S_T)\bigr)
  \notag \\
  &\quad= 1 - \exp\!\left(-\gamma\left(
      \frac{X_t}{\delta(t,T)} + \mathcal{G}^{l,m}
    \right)\right)
  \notag \\
  &\quad= 1 - e^{-\gamma X_t/\delta(t,T)} \cdot e^{-\gamma \mathcal{G}^{l,m}},
  \label{eq:utility_factored}
\end{align}
where the last step uses $e^{-(a+b)} = e^{-a} \cdot e^{-b}$. This is
the key step: the CARA utility converts the additive separation of $X_t$
in the terminal wealth into a multiplicative separation in the utility.
The factor $e^{-\gamma X_t/\delta(t,T)}$ is $\mathcal{F}_t$-measurable
and strictly positive, so it is a constant with respect to both the
conditional expectation and the supremum over future controls $(l,m)$.

Taking the supremum over $(l,m)$ and the conditional expectation,
\begin{align}
  U(t,S,y,X)
  &= \sup_{(l,m)} \mathbb{E}_{t,S,y,X}\!\left[
      1 - e^{-\gamma X_t/\delta(t,T)} \cdot e^{-\gamma \mathcal{G}^{l,m}}
    \right]
  \notag \\
  &= 1 - e^{-\gamma X/\delta(t,T)}
    \inf_{(l,m)} \mathbb{E}_{t,S,y}\!\left[e^{-\gamma \mathcal{G}^{l,m}}\right],
  \label{eq:U_factored}
\end{align}
where the $\sup$ over $(l,m)$ becomes an $\inf$ because maximising
$1 - e^{-\gamma X/\delta} \cdot Z$ over $(l,m)$ is equivalent to
minimising $Z = e^{-\gamma\mathcal{G}^{l,m}}$ over $(l,m)$, since the
prefactor $e^{-\gamma X/\delta} > 0$ is independent of the controls.
Setting
\begin{equation}
  Q(t, y, S) = \inf_{(l,m) \geq 0}
  \mathbb{E}_{t,y,S}\!\left[
    \exp\!\left(-\gamma\!\int_t^T
      \frac{(1-\lambda)S_s\,dm_s - (1+\lambda)S_s\,dl_s}{\delta(s,T)}
    \right)
    e^{-\gamma(y_T S_T - \lambda S_T|y_T| - \varphi(S_T))}
  \right],
\end{equation}
equation \eqref{eq:U_factored} gives the factorisation
$U(t,S,y,X) = 1 - e^{-\gamma X/\delta(t,T)} Q(t,y,S)$.
Since the infimum is taken over $(t,y,S)$ alone with no reference to
$X$, the function $Q$ does not depend on $X$.

\paragraph{\small Step 3: Derivation of the reduced quasi-variational inequality.}
We now substitute the factorised form $U = 1 - e^{-\gamma X/\delta}Q$
into the quasi-variational inequality $\max\{A, B, C\} = 0$ and derive the
equation satisfied by $Q$. The required partial derivatives are
\begin{align}
  U_X  &= \frac{\gamma}{\delta(t,T)}\, e^{-\gamma X/\delta}\, Q,
  \label{eq:UX} \\
  U_y  &= -e^{-\gamma X/\delta}\, Q_y,
  \label{eq:Uy} \\
  U_S  &= -e^{-\gamma X/\delta}\, Q_S,
  \label{eq:US} \\
  U_{SS} &= -e^{-\gamma X/\delta}\, Q_{SS},
  \label{eq:USS} \\
  U_t  &= -e^{-\gamma X/\delta}
           \!\left(Q_t - \frac{\gamma r X}{\delta(t,T)}\,Q\right).
  \label{eq:Ut}
\end{align}
Substituting into $A = U_y - (1+\lambda)S\,U_X$:
\begin{equation}
  A = -e^{-\gamma X/\delta}\!\left(
    Q_y + \frac{\gamma(1+\lambda)S}{\delta(t,T)}\,Q
  \right).
\end{equation}
Substituting into $B = -U_y + (1-\lambda)S\,U_X$:
\begin{equation}
  B = -e^{-\gamma X/\delta}\!\left(
    -Q_y - \frac{\gamma(1-\lambda)S}{\delta(t,T)}\,Q
  \right).
\end{equation}
Substituting into $C = -U_t - \mu S U_S - \frac{1}{2}\sigma^2 S^2
U_{SS} - rX U_X$, and noting that the $rXU_X$ term cancels the
$\gamma rX/\delta \cdot Q$ term arising from $U_t$:
\begin{equation}
  C = -e^{-\gamma X/\delta}\!\left(
    -Q_t - \mu S Q_S - \tfrac{1}{2}\sigma^2 S^2 Q_{SS}
  \right).
\end{equation}
In each case $A$, $B$, $C$ equal $-e^{-\gamma X/\delta}$ times a
function of $(t,y,S)$ alone. Since $-e^{-\gamma X/\delta} < 0$, the
condition $\max\{A, B, C\} = 0$ is equivalent to
\begin{equation}
  \min\!\left\{
    Q_y + \frac{\gamma(1+\lambda)S}{\delta(t,T)}\,Q,\;
    -Q_y - \frac{\gamma(1-\lambda)S}{\delta(t,T)}\,Q,\;
    -Q_t - \mu S Q_S - \tfrac{1}{2}\sigma^2 S^2 Q_{SS}
  \right\} = 0,
\end{equation}
which is the reduced quasi-variational inequality \eqref{eq:Q_vi}. The
terminal condition $Q(T,y,S) = \exp(-\gamma(yS - \lambda S|y| -
\varphi(S)))$ follows directly from \eqref{eq:terminal} and the
factorisation. This completes the proof of Proposition~\ref{prop:dim_reduction}.


\section{Sketch of the Whalley\texorpdfstring{\textendash}{}Wilmott Asymptotic Derivation}
\label{app:ww_derivation}

We outline the logic of the matched asymptotic expansion of \citet{WhalleyWilmott1997}, which reduces the three-dimensional free boundary problem for $Q$ to the closed-form bandwidth formula~\eqref{eq:ww_band}. The emphasis is on the structure of the argument rather than the algebraic details, which can be found in Section~3 of the original paper.

\subsection*{Setting}

The reduced variational inequality~\eqref{eq:Q_vi} partitions the $(S,y)$ plane at each time~$t$ into three regions: buy, sell, and no-transaction. In the buy and sell regions, the active conditions $A=0$ and $B=0$ are first-order in $y$, so for each fixed $(S,t)$ they can be integrated as ODEs in~$y$. This gives the \emph{exact} general solutions
\begin{align}
  Q^{\mathrm{buy}}  &= \exp\!\Bigl(-\tfrac{\gamma(1+\lambda) S\, y}{\delta} + H^-(S,t;\lambda)\Bigr), \label{eq:buy-sol}\\[3pt]
  Q^{\mathrm{sell}} &= \exp\!\Bigl(-\tfrac{\gamma(1-\lambda) S\, y}{\delta} + H^+(S,t;\lambda)\Bigr), \label{eq:sell-sol}
\end{align}
where $\delta = e^{-r(T-t)}$ and $H^\pm$ are undetermined functions of~$(S,t)$. Both solutions are \emph{exponential-linear} in~$y$: the $y$-dependence of $\log Q$ is exactly $-\gamma(1\pm\lambda)Sy/\delta$.

In the no-transaction region, $Q$ satisfies the full PDE $Q_t + \mu S Q_S + \tfrac{1}{2}\sigma^2 S^2 Q_{SS} = 0$, and must be joined to~\eqref{eq:buy-sol}\,--\,\eqref{eq:sell-sol} at the unknown free boundaries with $Q$, $Q_y$, and $Q_{yy}$ all continuous.

\subsection*{The expansion in the no-transaction region}

It is natural to work with $\log Q$ rather than $Q$ itself. Since the no-transaction region is a narrow band around the frictionless optimum $y^*(S,t)$, we introduce a rescaled variable $Y$ defined by $y = y^* + \lambda^{1/3} Y$, so that $Y = O(1)$ inside the band. The leading $y$-dependent term $-\gamma Sy^*/\delta$ in the expansion below is forced by continuity: both~\eqref{eq:buy-sol} and~\eqref{eq:sell-sol} contain $-\gamma Sy/\delta$ at $\lambda = 0$, and since $y = y^* + \lambda^{1/3}Y$ this splits into $-\gamma Sy^*/\delta$ at $O(1)$ and $-\lambda^{1/3}\gamma SY/\delta$ at $O(\lambda^{1/3})$. We then write
\[
  \log Q = -\frac{\gamma S\, y^*}{\delta} + H_0(S,t) \;-\; \lambda^{1/3}\frac{\gamma S Y}{\delta} \;+\; \lambda^{1/3} H_1(S,t) \;+\; \lambda^{2/3} H_2(S,t) \;+\; \lambda\, H_3(S,t) \;+\; \cdots
\]
where terms depending on both $Y$ and $(S,t)$ also appear at higher orders (starting at $O(\lambda^{4/3})$). The exponent $1/3$ is not chosen arbitrarily: it is the unique value for which the interior PDE and the boundary matching contribute at the same asymptotic order, as can be verified by tracking powers of $\lambda$ through the expansion (see \citet{WhalleyWilmott1997}, Section~3, for the detailed calculation). An economic interpretation of this scaling is given at the end of this appendix.

The functions $H_0, H_1, H_2, \ldots$ and the higher-order $Y$-dependent terms are determined by substituting the expansion into the no-transaction PDE and collecting powers of $\lambda^{1/3}$. A key subtlety is that the PDE is written in derivatives at fixed $y$, not fixed $Y$. Since $Y$ depends on $S$ and $t$ through $y^*$, the chain rule introduces factors of $y^*_S$ (the option's gamma) whenever $S$-derivatives act on $Y$-dependent terms, coupling different orders of the expansion.

\subsection*{What each order determines}

\noindent\textbf{$\boldsymbol{O(1)}$.}\quad
Determines $H_0(S,t)$ via a PDE relating it to $y^*$.

\medskip\noindent\textbf{$\boldsymbol{O(\lambda^{1/3})}$.}\quad
This order contains terms proportional to $Y$ and terms independent of $Y$; each group must vanish separately. The $Y$-dependent part yields the frictionless optimal holding:
$y^*(S,t) = \partial V_0/\partial S + \delta(\mu-r)/(\gamma S\sigma^2)$,
where $V_0 = (\delta/\gamma) H_0$ satisfies the Black\textendash Scholes equation.
The $Y$-independent part gives an equation for $H_1$, which turns out to be identically zero: the first non-trivial correction is at $O(\lambda^{2/3})$, not $O(\lambda^{1/3})$.

\medskip\noindent\textbf{$\boldsymbol{O(\lambda^{2/3})}$: the bandwidth.}\quad
This is the critical order. Through the chain-rule mechanism described above, $S$-derivatives acting on higher-order $Y$-dependent terms bring them down to this order: specifically, the second $S$-derivative of a $Y$-dependent function at $O(\lambda^{4/3})$ in $\log Q$ picks up a factor $\lambda^{-2/3}$ from two applications of $y^*_S \cdot \lambda^{-1/3}\partial/\partial Y$, contributing at $O(\lambda^{2/3})$ in the PDE. This yields an ODE in $Y$ with an inhomogeneous term proportional to $Y^2$, whose solution is a polynomial in~$Y$ with coefficients depending on $H_2(S,t)$ and on $y^*_S = \partial y^*/\partial S$.

The free constants and the band boundaries $Y^+$, $Y^-$ are pinned down by the matching conditions:
\begin{itemize}[leftmargin=2em]
\item \textbf{Continuity of $Q_y$}: the slope $\partial(\log Q)/\partial y$ must match the buy/sell solutions at each boundary.
\item \textbf{Continuity of $Q_{yy}$}: since $\log Q$ is linear in $y$ in the buy/sell regions, $(\log Q)_{yy}=0$ there, forcing the second $Y$-derivative of the polynomial to vanish at both boundaries.
\end{itemize}
These conditions together force (i) the band to be \emph{symmetric}: $Y^+ = Y^-$, and (ii) the half-width:
\begin{equation}\label{eq:Yplus}
  Y^+ = \left(\frac{3\, S\, \delta\, (y^*_S)^2}{2\gamma}\right)^{\!1/3}.
\end{equation}
The actual half-width in the original $y$-variable is $h = \lambda^{1/3} Y^+$.

\subsection*{The final formula}

With the option liability present and $\mu = r$, the speculative term in $y^*$ vanishes and $y^*_S = \Gamma_{\mathrm{BS}}(t,S)$. Substituting into~\eqref{eq:Yplus}:
\[
  h_{\mathrm{WW}}(t,S) = \left(\frac{3\,\lambda\, \delta(t,T)\, S\, \Gamma_{\mathrm{BS}}(t,S)^2}{2\gamma}\right)^{\!1/3},
\]
which is the formula~\eqref{eq:ww_band}.

\subsection*{Why $\lambda^{1/3}$: an economic interpretation}

The investor balances two costs. The \emph{suboptimality cost} of holding $y \neq y^*$ is quadratic in the deviation~$h$ (by concavity of the value function), contributing a running cost of order~$h^2$ per unit time. The \emph{transaction cost} of rebalancing is $\lambda$ per unit traded, but a wider band is hit less often: the expected cost per unit time is of order~$\lambda/h$. The total running cost is therefore $Ah^2 + B\lambda/h$, minimised at $h \propto \lambda^{1/3}$, with minimum value scaling as $\lambda^{2/3}$, which is why the leading correction to the option price is at order~$\lambda^{2/3}$, not~$\lambda$.

\begin{remark}
The full derivation, including the smoothness conditions, the explicit polynomial form of the $Y$-dependent terms, and the inhomogeneous Black\textendash Scholes equation satisfied by the $O(\lambda^{2/3})$ price correction, is given in Section~3 of \citet{WhalleyWilmott1997}.
\end{remark}
 

\section{Soft Clamp Gradient}
\label{app:softclamp}

We verify that the soft clamp \eqref{eq:softclamp} provides non-zero
gradients everywhere. Recall that
\begin{equation}
  \mathrm{softclamp}(x, \ell, u)
  = m + h \cdot \tanh\!\left(\beta \cdot \frac{x - m}{h}\right),
\end{equation}
where $m = (\ell + u)/2$ and $h = (u - \ell)/2 > 0$. The derivative
with respect to $x$ is
\begin{equation}
  \frac{\partial}{\partial x}\mathrm{softclamp}(x, \ell, u)
  = \beta \left(1 - \tanh^2\!\left(\beta \cdot
    \frac{x - m}{h}\right)\right)
  = \beta \operatorname{sech}^2\!\left(\beta \cdot
    \frac{x-m}{h}\right).
  \label{eq:softclamp_grad}
\end{equation}
Since $\operatorname{sech}^2(\cdot) > 0$ everywhere and $\beta > 0$,
the gradient \eqref{eq:softclamp_grad} is strictly positive for all
$x \in \mathbb{R}$. In particular, when $x \ll \ell$ or $x \gg u$,
the gradient is small but non-zero, of order
$4\beta e^{-2\beta|x-m|/h}$, providing a learning signal even when
the previous hedge is far outside the band.

By contrast, the hard clamp $\mathrm{clamp}(x, \ell, u)$ has zero
gradient for $x < \ell$ and $x > u$, which removes the gradient
signal for any trajectory where the position exits the band. As
$\beta \to \infty$, the soft clamp converges pointwise to the hard
clamp while maintaining non-zero gradients for all finite $\beta$.
In our experiments we use $\beta = 10$, which provides a close
approximation to the hard clamp within the band while preserving
meaningful gradient flow outside it.